\documentclass{article}
\pdfoutput=1
\usepackage[utf8]{inputenc}
\usepackage[T1]{fontenc}

\usepackage[margin=0.925in]{geometry}
\usepackage{amsmath,amsfonts, amssymb,amsthm}
\usepackage{bbm}
\usepackage{mathtools, thmtools}
\usepackage{xfrac}
\usepackage{indentfirst}
\usepackage{braket}
\usepackage{xcolor}
\usepackage[colorlinks,bookmarks=true]{hyperref}
\usepackage[capitalise,nameinlink]{cleveref}
\usepackage[shortlabels]{enumitem}
\usepackage{svg}
\usepackage{subcaption}
\usepackage{cancel}
\usepackage{array, booktabs, makecell}
\usepackage[normalem]{ulem}
\usepackage{tabularx, float} 
\usepackage{caption}
\usepackage{authblk}
\usepackage{hhline}
\usepackage{framed}
\usepackage{multicol}

\usepackage[backend=biber,
style=alphabetic,backref=true,maxnames=4,minnames=3,maxbibnames=99]{biblatex}
\DefineBibliographyStrings{english}{
  backrefpage={Cited on page},
  backrefpages={Cited on pages}
}

\setenumerate[1]{wide = 0pt,labelwidth = 0em, itemsep = 0em, leftmargin=*}
\newlist{defaultenumerate}{enumerate}{3}
\setlist[defaultenumerate,1]{label=\arabic*.}%
\setlist[defaultenumerate,2]{label=\arabic*.}%
\setlist[defaultenumerate,3]{label=\arabic*.}%

\allowdisplaybreaks[1]

\newcolumntype{L}[1]{>{\raggedright\arraybackslash}p{#1}}
\newcolumntype{C}[1]{>{\centering\arraybackslash}p{#1}}
\newcolumntype{R}[1]{>{\raggedleft\arraybackslash}p{#1}}

\newcommand{\calA}{\mathcal{A}}
\newcommand{\calB}{\mathcal{B}}
\newcommand{\calC}{\mathcal{C}}
\newcommand{\calD}{\mathcal{D}}
\newcommand{\calE}{\mathcal{E}}

\newcommand{\calP}{\mathcal{P}}

\newcommand{\sfS}{\mathsf{S}}
\newcommand{\Z}{\mathbb{Z}}
\newcommand{\F}{\mathbb{F}}
\newcommand{\1}{\mathbbm{1}}
\let\mod\undefined
\newcommand{\modulo}{\ \mathrm{mod} \ }
\newcommand{\modpm}{\ \mathrm{mod}^{\pm} \ }
\newcommand{\tr}{\mathrm{tr}}
\newcommand{\param}{\mathrm{par}}
\newcommand{\SIG}{\mathsf{SIG}}
\newcommand{\MSet}{\mathrm{MSet}}
\newcommand{\expect}{\mathrm{E}}
\DeclareMathSymbol{\mh}{\mathord}{operators}{`\-}
\newcommand{\floor}[1]{\lfloor#1 \rfloor}
\newcommand{\ceil}[1]{\lceil#1 \rceil}
\newcommand{\roundt}[1]{\lfloor#1 \rceil_t}
\newcommand{\abs}[1]{| #1 |}
\newcommand{\concat}{\ \| \ }
\newcommand{\norm}[1]{\| #1 \|}
\newcommand{\norminfty}[1]{\norm{#1}_{\infty}}
\newcommand{\mod}{\, \text{mod} \,}

\newcommand{\ketbra}[2]{\mathinner{|{#1}\rangle\hspace{-1pt}\langle{#2}|}}
\newcommand{\ketbrasame}[1]{\ketbra{#1}{#1}}

\newcommand{\Time}{\mathrm{Time}}
\newcommand{\polytext}{\mathrm{poly}}
\newcommand{\poly}{\mathrm{poly}(\lambda)}
\newcommand{\negl}{\mathrm{negl}(\lambda)}

\newcommand{\Dilithium}{\mathsf{Dilithium}}
\let\svhyphen-
\def\newhyphen{\raisebox{1pt}{\svhyphen}}
\newcommand{\DilithiumQROM}{\mathsf{Dilithium}\newhyphen\mathsf{QROM}}
\newcommand{\CRYSTALSDilithium}{\mathsf{CRYSTALS}\newhyphen\mathsf{Dilithium}}

\newcommand{\Adv}{\mathrm{Adv}}
\newcommand{\Collapse}{\mathsf{Collapse}}
\newcommand{\MLWE}{\mathsf{MLWE}}
\newcommand{\MSIS}{\mathsf{MSIS}}

\newcommand{\CCB}{\mathsf{CCB}}
\newcommand{\LWE}{\mathsf{LWE}}
\newcommand{\SIS}{\mathsf{SIS}}

\newcommand{\BKZ}{\mathsf{BKZ}}
\newcommand{\PR}{\mathsf{PR}}
\newcommand{\Sam}{\mathsf{Sam}}
\newcommand{\nHSTMSIS}{\mathsf{Plain}\newhyphen\mathsf{SelfTargetMSIS}}
\newcommand{\snHSTMSIS}{\mathsf{Plain}\raisebox{0pt}{\svhyphen}\mathsf{SelfTargetMSIS}}
\newcommand{\UFCMA}{\mathsf{EUF}\raisebox{0pt}{\svhyphen}\mathsf{CMA}}
\newcommand{\sUFCMA}{\mathsf{sEUF}\raisebox{0pt}{\svhyphen}\mathsf{CMA}}
\newcommand{\sEUFNMA}{\mathsf{sEUF}\raisebox{0pt}{\svhyphen}\mathsf{NMA}}
\newcommand{\STMSIS}{\mathsf{SelfTargetMSIS}}
\newcommand{\tMLWE}{\texorpdfstring{$\MLWE$}{mlwe} }

\newcommand{\tSTMSIS}{\texorpdfstring{$\STMSIS$}{selftargetmsis} }
\newcommand{\tCRYSTALSDilithium}{\texorpdfstring{$\CRYSTALSDilithium$}{CRYSTALSDilithium} }
\newcommand{\tDilithium}{\texorpdfstring{$\Dilithium$}{Dilithium} }
\newcommand{\tnHSTMSIS}{\texorpdfstring{$\nHSTMSIS$}{nHselftargetmsis} }
\newcommand{\tCCB}{\texorpdfstring{$\CCB$}{ccb} }
\newcommand{\tCollapse}{\texorpdfstring{$\Collapse$}{collapse} }
\newcommand{\KeyGen}{\mathrm{Gen}}
\newcommand{\Sign}{\mathrm{Sign}}
\newcommand{\Verify}{\mathrm{Verify}}

\newcommand{\Btau}{B_\tau}
\newcommand{\low}{\mathrm{low}}
\newcommand{\high}{\mathrm{high}}

\newtheorem{theorem}{Theorem}
\newtheorem{proposition}{Proposition}
\newtheorem{lemma}{Lemma}
\newtheorem{corollary}{Corollary}
\theoremstyle{definition}
\newtheorem{definition}{Definition}

\begin{document}

\title{Evaluating the security of \tCRYSTALSDilithium\\  in the quantum random oracle model}

\author[1]{Kelsey A.~Jackson\thanks{kaj22475@umd.edu}}
\author[1,2]{Carl A.~Miller\thanks{camiller@umd.edu}}
\author[1,3]{Daochen Wang\thanks{wdaochen@gmail.com}}
\affil[1]{\small Joint Center for Quantum Information and Computer Science (QuICS), University of Maryland}
\affil[2]{\small Computer Security Division, National Institute of Standards and Technology (NIST)}
\affil[3]{\small Department of Computer Science, University of British Columbia}

\date{\vspace{-6ex}}
\maketitle

\begin{abstract}
    In the wake of recent progress on quantum computing hardware, the National Institute of Standards and Technology (NIST) is standardizing cryptographic protocols that are resistant to attacks by quantum adversaries. The primary digital signature scheme that NIST has chosen is $\CRYSTALSDilithium$. The hardness of this scheme is based on the hardness of three computational problems: Module Learning with Errors ($\MLWE$), Module Short Integer Solution ($\MSIS$), and $\STMSIS$. $\MLWE$ and $\MSIS$ have been well-studied and are widely believed to be secure. However, $\STMSIS$ is novel and, though classically as hard as $\MSIS$, its quantum hardness is unclear. In this paper, we provide the first proof of the hardness of $\STMSIS$ via a reduction from $\MLWE$ in the Quantum Random Oracle Model (QROM). Our proof uses recently developed techniques in quantum reprogramming and rewinding. A central part of our approach is a proof that a certain hash function, derived from the $\MSIS$ problem, is collapsing. From this approach, we deduce a new security proof for $\Dilithium$ under appropriate parameter settings. Compared to the
    previous work by Kiltz, Lyubashevsky, and Schaffner (EUROCRYPT 2018) that gave the only other rigorous security proof for a variant of $\Dilithium$, our proof has the advantage of being applicable under the condition $q = 1 \modulo 2n$, where $q$ denotes the modulus and $n$ the dimension of the underlying algebraic ring.  This condition is part of the original $\Dilithium$ proposal and is crucial for the efficient implementation of the scheme. We provide new secure parameter sets for $\Dilithium$ under the condition $q = 1 \modulo 2n$, finding that our public key size and signature size are about $2.9\times$ and $1.3\times$ larger, respectively, than those proposed by Kiltz et al.~at the same security level.
\end{abstract}

\section{Introduction}

Quantum computers are theoretically capable of breaking the underlying computational hardness assumptions for many existing cryptographic schemes. Therefore, it is vitally important to develop new cryptographic primitives and protocols that are resistant to quantum attacks. 

The goal of NIST's Post-Quantum Cryptography Standardization Project is to design a new generation of cryptographic schemes that are secure against quantum adversaries.  In 2022, NIST selected three new digital signature schemes for standardization \cite{alagic2022status}: Falcon, SPHINCS+, and $\CRYSTALSDilithium$.  Of the three, $\CRYSTALSDilithium$ \cite{Dilithium}, or $\Dilithium$ in shorthand,  was identified as the primary choice for post-quantum digital signing.

To practically implement post-quantum cryptography, users must be provided with not only assurance that a scheme is secure in a post-quantum setting, but also the means by which to judge parameter choices and thereby balance their own needs for security and efficiency. The goal of the current work is to provide rigorous assurance of the security of $\Dilithium$ as well as implementable parameter sets. A common model for the security of digital signatures is existential unforgeability against chosen message attacks, or $\UFCMA$.  In this setting, an adversary is allowed to make sequential queries to a signing oracle for the signature scheme, and then afterwards the adversary attempts to forge a signature for a new message. We work in the setting of \textit{strong} existential unforgeability ($\sUFCMA$) wherein we must also guard against the possibility that an adversary could try to forge a new signature for one of the messages already signed by the oracle. (See \cref{sec:prelim} for details.)

Additionally, we utilize the quantum random oracle model (QROM) for hash functions.  We recall that when a hash function $H: X\rightarrow Y$ is used as a subroutine in a digital signature scheme, the random oracle model (ROM) assumes that one can replace each instance of the function $H$ with a black box that accepts inputs from $X$ and returns outputs in $Y$ according to a uniformly randomly chosen function from $X$ to $Y$. (This model is useful because random functions are easier to work with in theory than actual hash functions.)  The random oracle model needs to be refined in the quantum setting because queries to the hash function can be made in superposition: for any quantum state of the form $\sum_{x \in X} \alpha_x \left| x \right>$, where $\forall x\in X, \alpha_x \in \mathbb{C}$, a quantum computer can efficiently prepare the superposed state $\sum_{x \in X} \alpha_x \left| x \right> \left| H ( x ) \right>$. The quantum random oracle model (QROM) therefore assumes that each use of the hash function can be simulated by a black box that accepts a quantum state supported on $X$ and returns a quantum state supported on $X \times Y$ (computed by a truly random function from $X$ to $Y$) \cite{ROQW}. While no efficient and truly random functions actually exist, the QROM is generally trusted and it enables the application of a number of useful proof techniques.

\subsection{The \tDilithium signature scheme}

We give a brief description of $\CRYSTALSDilithium$.  (The reader is invited to consult \cite{Dilithium} for a full version of the protocol and a more detailed explanation of the design.) 
$\Dilithium$ is based on arithmetic over the ring $R_q \coloneqq \Z_q[X]/(X^n+1)$, where $q$ is an odd prime and $n$ is a power of $2$. Similar to other $\Dilithium$ literature, we generally leave the parameters $q,n$ implicit.  For any non-negative integer $\eta$, let $S_\eta \subseteq R_q$ denote the set of all polynomials with coefficients from $\{ - \eta, -\eta + 1 , \ldots, \eta \}$.  For any positive integer $\tau \leq n$, let $B_\tau \subseteq R_q$ denote the set of all polynomials $f$ such that exactly $\tau$ of the coefficients of $f$ are in $\{ -1, 1 \}$ and the remaining coefficients are all zero.

$\Dilithium$ is an instance of a general family of lattice-based signature schemes (see \cite[Subsection~5.6.2]{DecadePeikert}) that are obtained by applying the Fiat-Shamir transform to lattice-based interactive proofs-of-knowledge.
Neglecting some optimizations that are present in the full version of the scheme, we can concisely express $\Dilithium$ as in \cref{fig:dilithium_simplified}. The parameters $k, \ell, \gamma_1, \gamma_2, \tau, \beta$ are positive integers, and $H$ denotes a hash function which maps to the set $B_\tau$.  A signature for a message $M \in \{ 0, 1 \}^*$ takes the form of an ordered pair $\sigma = (z,c)$, where $z \in R_q^\ell$ and $c \in B_\tau$. 

\begin{figure}[ht]
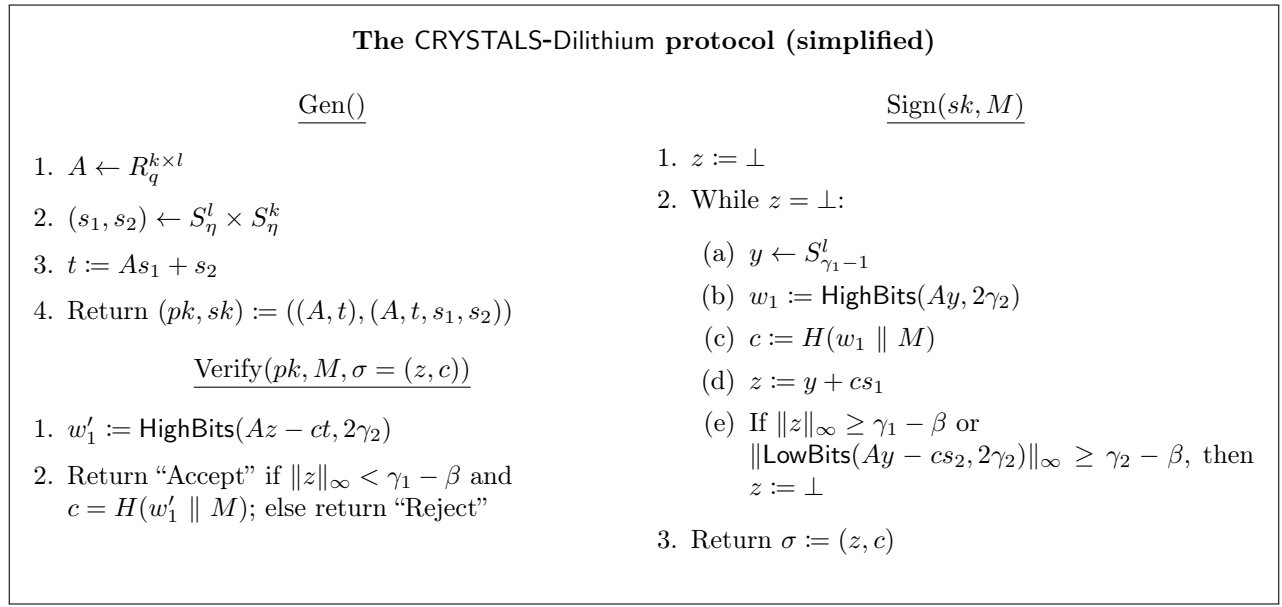

\centering
\begin{framed}
\begin{center} 
\textbf{The $\CRYSTALSDilithium$ protocol (simplified)} 
\begin{multicols}{2}
    \noindent \underline{$\KeyGen()$}
    \begin{enumerate}
    \item $A\leftarrow R_q^{k\times l}$
    \item $(s_1, s_2) \leftarrow S_\eta^l \times S_\eta^k$
    \item $t\coloneqq As_1 + s_2$
    \item Return $(pk, sk) \coloneqq ((A, t), (A, t, s_1, s_2))$
    \end{enumerate}
    \noindent \underline{$\Verify(pk, M, \sigma=(z, c))$}
    \begin{enumerate}
    \item $w_1' \coloneqq \mathsf{HighBits}(Az-ct, 2\gamma_2)$
    \item Return ``Accept'' if $\norminfty{z} < \gamma_1-\beta$ and \newline $c = H(w_1'\concat M)$; else return \mbox{``Reject''}
    \end{enumerate} 
    \pagebreak
    \noindent \underline{$\Sign(sk, M)$}
    \begin{enumerate}
    \item $z \coloneqq \bot$
    \item While $z = \bot$:
    \begin{enumerate}
        \item $y\leftarrow S_{\gamma_1-1}^{l}$
        \item $w_1 \coloneqq \mathsf{HighBits}(Ay, 2\gamma_2)$
        \item $c \coloneqq H(w_1 \concat M)$
        \item $z \coloneqq y + cs_1$
        \item If $\norminfty{z}\geq \gamma_1-\beta $ or \newline $ \norminfty{\mathsf{LowBits}(Ay-cs_2, 2\gamma_2)} \geq \gamma_2-\beta$,
        then $z \coloneqq \bot$
    \end{enumerate}
    \item Return $\sigma \coloneqq (z, c)$
    \end{enumerate}
\end{multicols}
\end{center}
\end{framed}
\caption{A simplified description of the key generation algorithm ($\KeyGen$), signature verification algorithm ($\Verify$) and signing algorithm ($\Sign$) for $\Dilithium$.}
\label{fig:dilithium_simplified}
\end{figure}

The algorithms in \cref{fig:dilithium_simplified} make use of the subroutines $\mathsf{HighBits}$ and $\mathsf{LowBits}$ which separate an $R_q$-vector into two parts. For any vector $x \in R_q^\ell$, the vectors $x_{\high} \coloneqq \mathsf{HighBits} ( x, 2 \gamma_2 )$ and 
$x_{\low} \coloneqq \mathsf{LowBits} ( x, 2 \gamma_2 )$ satisfy $x = (2 \gamma_2) x_{\high} + x_{\low}$, and the polynomial coefficients in $x_{\low}$ are all from the set $\{ -\gamma_2, -\gamma_2 + 1, \ldots, \gamma_2 \}$.

\subsection{Known security results for \tDilithium}
\label{subsec:knownsecurity}

The security analysis for $\Dilithium$ in \cite{Dilithium} is based on three computational problems.  The first two are standard problems (\cref{prob:MLWE,prob:MSIS}) but the third problem is non-standard (\cref{prob:STMSIS}).  The first problem is the Module Learning With Errors ($\MLWE$) problem. Assuming that a matrix $A \in R_q^{m \times k}$ and short vectors $s_1 \in S_\eta^k$ and $s_2 \in S_\eta^m$ are chosen uniformly at random,
the $\MLWE$ problem is to distinguish the matrix-vector pair $(A, t \coloneqq As_1 + s_2)$ from a  uniformly random matrix-vector pair.

\begin{definition}[Module Learning with Errors ($\MLWE$)]\label{prob:MLWE} 
Let $m, k, \eta \in \mathbb{N}$. The advantage of an algorithm $\calA$ for solving $\MLWE_{m, k, \eta}$ is defined as:
\begin{align}\label{eq:mlwe}
\begin{aligned}
    \Adv^{\MLWE}_{m,k,\eta}(\calA) \coloneqq& \bigl\vert \Pr[b=0 \mid A \leftarrow R_q^{m\times k},  \ t\leftarrow R_q^m, \ b\leftarrow\calA(A, t)] 
    \\
  &\quad- \Pr[b=0 \mid A \leftarrow R_q^{m\times k}, \ (s_1, s_2) \leftarrow S_\eta^k \times S_\eta^m, t \coloneqq As_1 + s_2, \ b\leftarrow\calA(A, t)] \bigr\vert.
\end{aligned}
\end{align}
\end{definition}

Here, the notation $\calA(x)$ denotes $\calA$ taking input $x$.
We note that the $\MLWE$ problem is often phrased in other contexts with the short vectors $s_1$ and $s_2$ coming from a Gaussian, rather than a uniform, distribution.  The use of a uniform distribution is one of the particular features of $\CRYSTALSDilithium$.

The second problem, $\MSIS$, is concerned with finding short solutions to randomly chosen linear systems over $R_q$.

\begin{definition}[Module Short Integer Solution ($\MSIS$)]\label{prob:MSIS}
Let $m,k, \gamma \in \mathbb{N}$. The advantage of an algorithm $\calA$ for solving $\MSIS_{m, k, \gamma}$ is defined as:
\begin{equation}
\begin{aligned}
     \Adv^{\MSIS}_{m,k,\gamma}(\calA) \coloneqq \Pr\bigl[ [I_m|A]  \cdot & y = 0 \wedge 0 < \norminfty{y} \leq \gamma \mid  A\leftarrow R_q^{m\times k}, \ y\leftarrow\calA(A)\bigr].
\end{aligned} 
\end{equation}
\end{definition}

The third problem is a more complex variant of $\MSIS$ that incorporates a hash function $H$.

\begin{definition}[$\STMSIS$]\label{prob:STMSIS}
Let $\tau,m,k,\gamma \in \mathbb{N}$ and $H \colon \{0,1\}^* \to \Btau$, where $\Btau \subseteq R_q$ is the set of polynomials with exactly $\tau$ coefficients in $\{-1,1\}$ and all remaining coefficients zero. The advantage of an algorithm $\calA$ for solving $\STMSIS_{H,\tau,m,k,\gamma}$ is defined as\footnote{$\concat$ denotes string concatenation. $\calA^{\ket{H}}$ denotes $\calA$ with quantum query access to $H$ --- a formal definition can be found in \cref{def:qAlg}.}:
\begin{equation}
\Adv^{\STMSIS}_{H,\tau,m,k,\gamma}(\calA) \coloneqq \Pr\Bigl[H( [I_m | A] \cdot y  \concat M) = y_{m+k} \wedge \norminfty{y} \leq \gamma \bigm| \\
A \leftarrow R_q^{m\times k}, \ (y, M) \leftarrow\calA^{\ket{H}}(A) \Bigr].
\end{equation}
\end{definition}

The security guarantee for $\CRYSTALSDilithium$ is given in \cite[Section 4.5]{Concrete} by the inequality\footnote{Strictly speaking, there should be two other terms ($\Adv_{\Sam}^\PR(\calE)$ and $2^{-\alpha+1}$) on the right-hand side of \cref{eq:dilithium_security}. However, we ignore them in the introduction as it is easy to set parameters such that these terms are very small. We also mention that the original proof of this inequality uses a flawed analysis of Fiat-Shamir with aborts. The flaw was found and fixed in \cite{FixingWu,FixingStehle}.}
\begin{align}\label{eq:dilithium_security}
    \Adv^{\sUFCMA}_{\Dilithium}(\calA) \leq  \Adv_{k, l, \eta}^{\MLWE}(\calB) + \Adv_{H, \tau, k, l+1, \zeta}^{\STMSIS}(\calC) + \Adv_{k, l, \zeta'}^{\MSIS}(\calD),
\end{align}
where all terms on the right-hand side of the inequality depend on parameters that specify $\Dilithium$, and $\sUFCMA$ stands for strong unforgeability under chosen message attacks.
The interpretation of \cref{eq:dilithium_security} is: if there exists a quantum algorithm $\calA$ that attacks the $\sUFCMA$-security of $\Dilithium$, then there exist quantum algorithms $\calB,\calC,\calD$ for $\MLWE$, $\STMSIS$, and $\MSIS$ that have advantages satisfying \cref{eq:dilithium_security} and run in time comparable to $\calA$. \cref{eq:dilithium_security} implies that breaking the $\sUFCMA$ security of $\Dilithium$ is at least as hard as solving one of the $\MLWE$, $\MSIS$, or $\STMSIS$ problems. $\MLWE$ and $\MSIS$ are known to be no harder than $\LWE$ and $\SIS$, respectively. However, there are no known attacks taking advantage of their module structure so it is generally believed that they are as hard as their unstructured counterparts \cite{langloisStehle}. In turn, $\LWE$ and $\SIS$ are at least as hard as the (Gap) Shortest Vector Problem, which is the underlying hard problem of lattice cryptography \cite{SisAtjai,RegevLwe,DecadePeikert}.

However, the final problem, $\STMSIS$, is novel and so its difficulty is an open question. The problem is known to be as classically hard as $\MSIS$ since there exists a reduction from $\MSIS$ to $\STMSIS$ in the ROM \cite{Concrete,Forking}. The reduction uses the following ``rewinding'' argument. Any randomized algorithm can be specified by a deterministic circuit with auxiliary random bits. Therefore, given a randomized algorithm for $\STMSIS$, we can run its deterministic circuit with some randomly chosen bits to obtain one solution and then rewind and run it again using the same bits chosen from before, while at the same time reprogramming the random oracle at the query corresponding to the output of the first run, to obtain a second solution. Subtracting these two solutions to $\STMSIS$ yields a solution to $\MSIS$.
However, the argument fails for the following reasons in the QROM (where a quantum algorithm can make queries in superposition to a quantum random oracle):
\begin{enumerate}
    \item The randomness in a quantum algorithm includes the randomness of measurement outcomes. We cannot run a quantum algorithm twice and guarantee that the ``random bits'' will be the same in both runs because we cannot control measurement outcomes. More generally, we cannot rewind a quantum algorithm to a post-measurement state.
    \item Since a quantum algorithm can make queries in superposition, it is no longer clear where to reprogram the random oracle.
\end{enumerate}

Currently, the only explicit rigorous proof of $\Dilithium$'s security based on conventional hardness assumptions 
\cite{Concrete} requires modifying the parameters to be such that $q = 5 \mod 8$ and $2\gamma < \sqrt{q/2}$ (but $n$ must remain a power of $2$), where $\gamma$ is a length upper bound on vectors corresponding to valid signatures. This ensures that all non-zero vectors in $S_{2\gamma}$ are invertible which equips $\Dilithium$ with a so-called ``lossy mode''. This variant is called $\DilithiumQROM$. \cite{Concrete} then proves that a signature scheme with such a lossy mode is $\UFCMA$. However, the $\Dilithium$ specification~\cite{Dilithium} uses a value of $q$ satisfying $q = 1 \mod 2n$ which is incompatible with the assumption that $q = 5 \mod 8$ and $n>2$ is a power of $2$.\footnote{The parameter $n$ should not be $1$ or $2$, as that would significantly degrade Dilithium's efficiency and defeat the purpose of its use of $R_q$.} The fact that $q = 1 \mod 2n$ is central to claims about the speed of the algorithms in \cite{Dilithium}: this condition implies that $R_q$ is isomorphic to the direct product ring $\Z_q^{\times n}$ (or $\Z_q^n$ in shorthand) via the Number Theoretic Transform (NTT), which allows for fast matrix multiplication over $R_q$. Therefore, it is highly desirable to find a security proof that works under the assumption that $q = 1 \mod 2n$. Moreover, when $q = 5 \mod 8$ (and $n$ is a power of $2$), the ring $R_q$ is \emph{structurally} different from when $q = 1 \mod 2n$, since in the former case $R_q$ is isomorphic to $\F_{q^{n/2}}\times \F_{q^{n/2}}$ \cite[Lemma 2.1]{VerifiableEncryptionLyubashevskyNeven}. Therefore, it may be imprudent to translate any claims of security in the case $q = 5 \mod 8$ to the case $q = 1 \mod 2n$.

\subsection{Overview of main result}

The main result of our paper is the first proof of the computational hardness of the $\STMSIS$ problem, presented in \cref{sec:security_stmsis}. 
This hardness result implies a new security proof for $\Dilithium$ which, unlike the previous proof in \cite{Concrete}, applies
to the case $q=1 \mod 2n$.
Specifically, we reduce $\MLWE$ to $\STMSIS$.  
By \cref{eq:dilithium_security}, our result implies that the security of $\Dilithium$ (with parameters that are not too far from the original parameters) can be based on the hardness of $\MLWE$ and $\MSIS$.

\begin{theorem}[{Informal version of \cref{thm:stmsis}}]
\label{thm:stmsisinformal}
    Let $m,k,\tau,\gamma,\eta \in \mathbb{N}$. Suppose $q \geq 16$, $q = 1 \modulo 2n$, and $2\gamma\eta n(m+k) < \floor{q/32}$. If there exists an efficient quantum algorithm $\mathcal{A}$ that solves $\STMSIS_{H,\tau, m,k,\gamma}$ with advantage $\epsilon$, under the assumption that $H$ is a random oracle, then there exists an efficient quantum algorithm for solving $\MLWE_{m+k,m,\eta}$ with advantage at least $\Omega(\epsilon^2/Q^4)$.  Here, $Q$ denotes the number of quantum queries $\mathcal{A}$ makes to $H$.
\end{theorem}

We now give a high-level overview of the proof. The first step is to define two experiments: the \emph{chosen-coordinate binding} experiment $\CCB$ and the \emph{collapsing} experiment $\Collapse$. These experiments are interactive protocols between a verifier and a prover. The protocols end with the verifier outputting a bit $b$. If $b=1$, the prover is said to \emph{win} the experiment. The reduction then proceeds in three steps: (i) reduce winning $\CCB$ to solving $\STMSIS$, (ii) reduce winning $\Collapse$ to winning $\CCB$, and (iii) reduce solving $\MLWE$ to winning $\Collapse$. Combining these steps together gives a reduction from $\MLWE$ to $\STMSIS$. The reduction can be illustrated as
\begin{align}
    \STMSIS \overset{(i)}{\longleftarrow} \CCB \overset{(ii)}{\longleftarrow} \Collapse \overset{(iii)}{\longleftarrow} \MLWE,
\end{align}
where the left arrow means ``reduces to''.

\paragraph{Step (i): $ \STMSIS \leftarrow \CCB$.}   

In the $\CCB$ experiment, the prover is first given a uniformly random $A\in R_q^{m\times l}$ which it uses to send the verifier some $z \in R_q^m$, the verifier then sends the prover a challenge $c$ chosen uniformly at random from $\Btau$, and finally the prover sends the verifier a response $y \in R_q^l$. The prover wins if $Ay = z$, $\norminfty{y}\leq \gamma$, and the last coordinate of $y$ is $c$. 

We directly apply the main result of \cite{Reprogram} to reduce winning $\CCB$ when $l=m+k$ to solving $\STMSIS_{H,\tau,m,k,\gamma}$ when $H$ is a random oracle. In more detail, the result implies that an efficient algorithm that wins $\STMSIS$ using $Q$ queries with probability $\epsilon$ can be used to construct another efficient algorithm that wins $\CCB$ with probability at least $\Omega(\epsilon/Q^2)$.

\paragraph{Step (ii): $\CCB \leftarrow \Collapse$.} In the $\Collapse$ experiment, the prover is first given a uniformly random $A\in R_q^{m\times l}$ which it uses to send the verifier some $z\in R_q^m$ together with a quantum state that must be supported only on $y \in R_q^l$ such that $Ay = z$, $\norminfty{y} \leq \gamma$. Then, the verifier samples a uniformly random bit $b'$. If $b'=1$, the verifier measures the quantum state in the computational basis, otherwise, it does nothing. The verifier then returns the quantum state to the prover. The prover responds by sending a bit $b'$ to the verifier and wins if $b'=b$. The advantage of the prover is $2p-1$ where $p$ is its winning probability.

By using techniques in \cite{NecessityCollapsing,CollapsingWithoutRoUnruh}, we reduce winning $\Collapse$ to winning $\CCB$. More specifically, we show that an efficient algorithm that wins $\CCB$ with advantage $\epsilon$ can be used to construct another efficient algorithm that wins $\Collapse$ with advantage at least $\epsilon(\epsilon-1/\abs{\Btau})$, which is roughly $\epsilon^2$ since $1/\abs{\Btau}$ is very small for the values of $\tau$ we will consider. We generalize techniques in \cite{NecessityCollapsing,CollapsingWithoutRoUnruh} to work for challenge sets of size $>2$, which is necessary since the challenge set in the $\CCB$ experiment, $\Btau$, generally has size $>2$. The key idea of first applying the quantum algorithm for winning $\CCB$ to the uniform superposition of all challenges remains the same.

\paragraph{Step (iii): $\Collapse \leftarrow \MLWE$.} We build on techniques in \cite{QuantumMoneyLiuMontgomeryZhandry,FiatShamirLiuZhandry} to reduce winning $\Collapse$ to winning $\MLWE$. More specifically, we show that an efficient algorithm that wins $\Collapse$ with advantage $\epsilon$ can be used to construct another efficient algorithm that solves $\MLWE_{l,m,\eta}$ with advantage at least $\epsilon/4$. Given a quantum state supported on $y\in R_q^l$ with $Ay = z$ and $\norminfty{y} \leq \gamma$, as promised in the $\Collapse$ experiment, \cite{QuantumMoneyLiuMontgomeryZhandry,FiatShamirLiuZhandry} considers the following two measurements. Sample $b\in R_q^l$ from one of the two distributions defined in $\MLWE$ (see \cref{eq:mlwe}), compute a rounded version of $b\cdot y$ in a separate register, and measure that register. When $n=1$, \cite{QuantumMoneyLiuMontgomeryZhandry} shows that the effect of the measurement in one case is close to the computational basis measurement and in the other case is close to doing nothing. Therefore, an algorithm for winning $\Collapse$ can be used to solve $\MLWE$. Our work extends \cite{QuantumMoneyLiuMontgomeryZhandry} to arbitrary $n$ provided $q = 1 \mod 2n$. The extension relies on the fact that each coefficient of $b\cdot \Delta$, where $0\neq \Delta \in R_q$ and $b$ is chosen uniformly at random from $R_q$, is uniformly random in $\mathbb{Z}_q$. (This is despite the fact that $b\cdot \Delta$ is generally not uniformly random in $R_q$.) We establish this fact using the explicit form of the isomorphism between $R_q$ and $\mathbb{Z}_q^n$ when $q = 1 \mod 2n$.

Finally, in \cref{sec:concrete}, we propose explicit sets of parameters using $n=512$ and $q \approx 2^{43.5}$ such that $q = 1\mod 2n$. These sets of parameters achieve different security levels based on \cref{thm:stmsisinformal}. We compare our sets of parameters with sets proposed by the $\Dilithium$ specifications \cite{Dilithium} and the $\DilithiumQROM$ construction of \cite{Concrete}. We also compute the differences in the number of $\mathbb{Z}_q$-operations caused by using the NTT on $q=1\mod 2n$ for ring element multiplication compared to a Hybrid-NTT on $q=5\mod 8$.  We find that our public key and signatures sizes are $\approx 11.4\times$ and $\approx 3.2\times$ larger, respectively, than the heuristically chosen parameters in the original $\Dilithium$ \cite{Dilithium}.  Compared to \cite{Concrete}, our parameters yield an increase in 
public key size and signature size of 
$\approx 2.9\times$ and $\approx 1.3\times$, respectively, while yielding a significant decrease (because of the different structure of $R_q$) in the number of $\mathbb{Z}_q$-operations used in key generation, signing, and verification.

\cref{thm:stmsisinformal} proves security for $\Dilithium$ in a range of parameters that preserves the algebraic structure of the original protocol \cite{Dilithium}.  Future work could explore how to optimize our approach to obtain provably secure parameters that are closer to those proposed in \cite{Dilithium} for practical use. We also seek to generalize these results to other signature schemes that utilize the Fiat-Shamir transform.

\section{Preliminaries}
\label{sec:prelim}

$\mathbb{N}$ denotes the set of positive integers. For $k \in 
\mathbb{N}$, $[k]$ denotes the set $\{1,\dots,k\}$. An alphabet refers to a finite non-empty set. Given an alphabet $S$, the notation $s\leftarrow S$ denotes selecting an element $s$ uniformly at random from $S$. Given two alphabets $A$ and $B$, the notation $B^A$ denotes the set of functions from $A$ to $B$. We write the concatenation of arbitrary strings $a, b$ as $a \concat b$. Given matrices $A_1,\dots,A_n$ of the same height, $\left[A_1 \vert A_2 \vert ... \vert A_n \right]$ denotes the matrix that consists of the $A_i$ matrices placed side by side. $\log$ refers to the base-$2$ logarithm.

We always reserve the symbol $q$ for an odd prime and $n$ for a positive integer that is a power of $2$. $R_{q}$ denotes the ring $\Z_q[X]/(X^n+1)$ (following the convention in other $\Dilithium$ literature \cite{Dilithium,Concrete}, we leave the $n$-dependence implicit). For $k\in \mathbb{N}$, a primitive $k$th root of unity in $\mathbb{Z}_q$ is an element $x\in \mathbb{Z}_q$ such that $x^k = 1$ and $x^j \neq 1$ for all $j\in [k-1]$; such elements exist if and only if $q = 1 \modulo k$. Given $r\in \mathbb{Z}_q$, we define $r \modpm q$ to be the unique element $r'\in \mathbb{Z}$ such that $-(q-1)/2\leq r'\leq (q-1)/2$ and $r' = r \mod q$. For any $r = a_0 + a_1 X + \dots + a_{n-1}X^{n-1}\in R_q$, we define $\abs{r}_i \coloneqq \abs{a_i \modpm q}$ for all $i\in \{0,1,\dots,n-1\}$ and $\norm{r}_\infty \coloneqq \max_i \abs{r}_i$. For $r \in R_q^m$, we define $\norm{r}_\infty \coloneqq \max_{i\in [m]} \norm{r_i}_{\infty}$. For $\eta\in \mathbb{N}$, $S_\eta$ denotes the set $\{r \in R_q \mid \norminfty{r}\leq \eta\}$. For $\tau \in \mathbb{N}$, $\Btau$ denotes the set of all elements $r \in R_q$ such that $\norminfty{r} = 1$ and $r$ has exactly $\tau$ nonzero coefficients.
 We note that  $\abs{\Btau} = 2^\tau \binom{n}{\tau}$.

\subsection{Quantum computation}

A (quantum) state, or density matrix, $\rho$ on $\mathbb{C}^d$ is a positive semi-definite matrix in $\mathbb{C}^{d\times d}$ with trace $1$. A pure state is a state of rank $1$. Since a pure state can be uniquely written as $\ketbrasame{\psi}$ where $\ket{\psi}\in \mathbb{C}^d$ and $\bra{\psi}\coloneqq \ket{\psi}^\dagger$, we usually refer to a pure state by just $\ket{\psi}$. A (projective) measurement is a set $\calP = \{P_1,\dots,P_k\} \subseteq \mathbb{C}^{d\times d}$ such that $\sum_{i=1}^k P_i = 1$ and the operators $P_i$ satisfy $P_i = P_i^\dagger = P_i^2$ and $P_i P_j = 0$ for any $j \neq i$. The effect of performing such a measurement on a quantum state $\rho$ is to produce the density matrix $\sum_{i=1}^k P_i\rho P_i$. 

A register is either (i) an alphabet $\Sigma$ or (ii) an $m$-tuple $X = (Y_1,\dots,Y_m)$ where $m\in \mathbb{N}$ and $Y_1,\dots,Y_m$ are alphabets. 
\begin{enumerate}[label=Case (\roman*)]
    \item The size of the register is $\abs{\Sigma}$, a density matrix on the register refers to a density matrix on $\mathbb{C}^{\abs{\Sigma}}$, and the computational basis measurement on the register refers to the measurement $\{\ketbrasame{x}\mid x \in \Sigma\}$, where $\ket{x}$ denotes the vector in $\mathbb{C}^{\Sigma} \cong \mathbb{C}^{\abs{\Sigma}}$ that is $1$ in the $x$th position and zero elsewhere.
    \item The size of the register is $\abs{Y_1}\times \dots \times \abs{Y_m}$, a density matrix on the register refers to a density matrix on $\mathbb{C}^{\abs{Y_1}}\otimes \dots \otimes \mathbb{C}^{\abs{Y_m}}$, and the computational basis measurement on the register refers to the measurement $\{\ketbrasame{y_1}\otimes \dots\otimes \ketbrasame{y_m} \mid y_1\in Y_1,\dots, y_m \in Y_m\}$.
\end{enumerate}

A quantum algorithm $\calA$ is specified by a register $X = (Y_1,\dots,Y_m)$ where $\abs{Y_i} = 2$ for all $i$ and a sequence of elementary gates, i.e., $2^m\times 2^m$ unitary matrices that are of the form
\begin{align}
\begin{split}
    &T\coloneqq \begin{pmatrix}
       1 & 0
       \\
       0 & e^{i\pi/4}
    \end{pmatrix},  \quad 
    H \coloneqq \frac{1}{\sqrt{2}}\begin{pmatrix}
       1 & 1
       \\
       1 & -1
    \end{pmatrix}, \quad \text{or} \quad
    \text{CNOT} \coloneqq \begin{pmatrix}
        1 & 0 & 0 & 0
        \\
        0 & 1 & 0 & 0
        \\
        0 & 0 & 0 & 1
        \\
        0 & 0 & 1 & 0
    \end{pmatrix}
\end{split}
\end{align}
tensored with $2\times 2$ identity matrices.\footnote{When we later consider a quantum algorithm on a register of size $d \in \mathbb{N}$, we mean a quantum algorithm on a register $(Y_1,\dots, Y_m)$ where $\abs{Y_i} = 2$ for all $i$ and $m$ is the smallest integer such that $2^m \geq d$.} The unitary matrix $U$ associated with $\calA$ is the product of its elementary gates in sequence. The time complexity of $\calA$, $\Time(\calA)$, is its number of elementary gates. To perform a computation given an input $x\in \{0,1\}^k$ where $k\leq m$, $\calA$ applies $U$ to the starting state $\ket{\psi_0} \coloneqq \ket{x_1+1}\otimes \dots \otimes \ket{x_k+1}\otimes \ket{1}^{\otimes (m-k)}$ and measures all registers in the computational basis. We also need the definition of a quantum \emph{query} algorithm.

\begin{definition}[Quantum query algorithm]\label{def:qAlg} Let $t\in \mathbb{N}$. A quantum query algorithm $\calA$ using $t$ queries is specified by registers $X,Y,Z$ and a sequence of $t+1$ quantum algorithms $\calA_0,\calA_1,\dots,\calA_t$, each with register $(X,Y,Z)$. The time complexity of $\calA$, $\Time(\calA)$, is $t+\sum_{i=0}^t\Time(\calA_i)$. 

Let $U_i$ denote the unitary associated with $\calA_i$, $\gamma \coloneqq \abs{Y}$, and $\phi\colon Y \to \mathbb{Z}_\gamma$ be a bijection. Given $H \colon X\to Y$, let $O^H$ denote the unitary matrix defined by $O^H\ket{x}\ket{y}\ket{z} = \ket{x}\ket{\phi^{-1}(\phi(y) + \phi(H(x)))} \ket{z}$ for all $(x,y,z)\in X\times Y\times Z$. Then:
\begin{enumerate}
    \item $\calA^{\ket{H}}$ denotes the algorithm with register $(X,Y,Z)$ that computes as follows. Apply $U_0$ to the starting state $\ket{\psi_0}$. Then, for each $i = 1,\dots,t$ in sequence, apply $O^H$ then $U_i$. Finally, measure all registers in the computational basis. 
    \item $\calA^{H}$ denotes the algorithm with register $(X,Y,Z)$ that computes as follows. Apply $U_0$ to the starting state $\ket{\psi_0}$. Then, for each $i = 1,\dots,t$ in sequence, measure register $X$ in the computational basis and apply $O^H$ then $U_i$. Finally, measure all registers in the computational basis. 
\end{enumerate} 
\end{definition}
In the definitions of $\calA^{\ket{H}}$ and $\calA^H$, we have described what it means for a quantum algorithm to make quantum and classical queries to a function $H$, respectively. Under this description, we can naturally define quantum query algorithms that make classical queries to one function and quantum queries to another. Such algorithms are relevant in the security definition of $\Dilithium$ as described in the next subsection.

\subsection{Digital signature schemes}

Let $\param$ be common system parameters shared by all participants. For $\lambda \in \mathbb{N}$ we define $\negl$ as some function such that $\negl \leq 1/\eta^c$ for all constants $c$ and large enough values of $\lambda$.  
\begin{definition}[Digital signature scheme\label{def:sigScheme}] A digital signature scheme is defined by a triple of randomized algorithms $\SIG = (\KeyGen, \Sign, \Verify)$ such that
\begin{enumerate}
    \item The key generation algorithm $\KeyGen(\param)$ outputs a public-key, secret-key pair $(pk,sk)$ such that $pk$ defines the message set $\MSet$.
    \item The signing algorithm $\Sign(sk,m)$, where $m\in \MSet$, outputs a signature $\sigma$.
    \item The verification algorithm $\Verify(pk,m,\sigma)$ outputs a single bit $\{0,1\}$. 
\end{enumerate}
We say $\SIG$ has correctness error $\gamma\geq 0$ if for all $(pk,sk)$ in the support of $\KeyGen(\param)$ and all $m\in \MSet$, 
\begin{align}
    \Pr[\Verify(pk, m, \sigma) = 0 \mid \sigma\leftarrow \Sign(sk, m)] \leq \gamma.
\end{align}
\end{definition}

\begin{definition}[$\UFCMA$ and $\sUFCMA$]\label{def:UF} Let $\SIG = (\KeyGen, \Sign, \Verify)$ be a signature scheme. Let $\calA$ be a quantum query algorithm. Then
\begin{align*}
    \Adv_{\SIG}^{\UFCMA}(\calA) \coloneqq& \Pr[\Verify(pk, m, \sigma) = 1, \ m \notin \textup{SignQ} \mid (pk, sk) \leftarrow  \KeyGen(\param), (m, \sigma)\leftarrow\calA^{\Sign(sk, \cdot)}(pk)],
\\
    \Adv_{\SIG}^{\sUFCMA}(\calA) \coloneqq& \Pr[\Verify(pk, m, \sigma) = 1, \ (m, \sigma)\notin \textup{SignQR} \mid (pk, sk) \leftarrow  \KeyGen(\param), (m, \sigma)\leftarrow\calA^{\Sign(sk, \cdot)}(pk)],
\end{align*}
where $\textup{SignQ}$ is the set of queries made by $\calA$ to $\Sign(sk, \cdot)$ and  $\textup{SignQR}$ is the set of query-response pairs $\calA$ sent to and received from $\Sign(sk,\cdot)$.
\end{definition}
When $\param$ is a function of $\lambda\in \mathbb{N}$, we say that $\SIG$ is $\textrm{(s)}\UFCMA$-secure if for every $\poly$-time quantum query algorithm $\calA$, we have $\Adv_{\SIG}^{\textrm{(s)}\UFCMA}(\calA) \leq \negl$.

In this paper, we use the definition of the $\Dilithium$ signature scheme as specified in \cite{Dilithium}. In the concrete parameters section, \cref{sec:concrete}, we adopt the same notation as in \cite{Dilithium}. The definition of $\Dilithium$ involves a function $H\colon \{0,1\}^*\to \Btau$ that is classically accessible by its $\Sign$ and $\Verify$ algorithms. In the definitions of $\UFCMA$ and $\sUFCMA$ security of $\Dilithium$, we assume that the quantum algorithm $\calA$ has classical query access to $\Sign(sk,\cdot)$ and quantum query access to $H$. Our proof of $\Dilithium$'s security will assume that $H$ can be modeled by a random oracle. 

\subsection{Cryptographic problems and experiments}

We now give the formal definitions of the chosen-coordinate binding and collapsing experiments mentioned in the introduction.  More general versions of these definitions can be found in, e.g., \cite{UnruhStrict,NecessityCollapsing}.

In \cref{subsec:knownsecurity} we defined the $\STMSIS$ problem (\cref{prob:STMSIS}).
Now we define a ``plain'' version of $\STMSIS$, where the input matrix is not given in Hermite Normal Form. First reducing $\STMSIS$ from $\nHSTMSIS$ will be convenient later on.

\begin{definition}[$\nHSTMSIS$]\label{prob:nHSTMSIS}
Let $\tau, m,l,\gamma \in \mathbb{N}$ and $H \colon \{0,1\}^* \to B_\tau$. The advantage of a quantum query algorithm $\calA$ for solving \newline $\nHSTMSIS_{H,\tau,m,l,\gamma}$ is defined as
\begin{equation}
\Adv^{\snHSTMSIS}_{H,\tau,m,l,\gamma}(\calA) \coloneqq \Pr\Bigl[H( A y  \concat M) = y_l \wedge \norminfty{y} \leq \gamma \bigm|A\leftarrow R_q^{m\times l}, \ (y, M) \leftarrow\calA^{\ket{H}}(A) \Bigr].
\end{equation}
\end{definition}

\begin{definition}[Chosen-Coordinate Binding ($\CCB$)]\label{def:ccb}
Let $\tau, m,l,\gamma \in \mathbb{N}$. The advantage of a quantum algorithm $\calA = (\calA_1,\calA_2)$ for winning $\CCB_{\tau,m,l,\gamma}$, denoted $\Adv^{\CCB}_{\tau,m,k,\gamma}(\calA)$, is defined as the probability that the experiment below outputs $1$.
\begin{center}
\fbox{\begin{minipage}{0.8\textwidth}
\textbf{\textup{Experiment}} $\CCB_{\tau,m,l,\gamma}$. 
\begin{enumerate}[topsep=3pt]
    \item Sample $A \leftarrow R_q^{m\times l}$.
    \item $(z,T) \leftarrow\calA_1(A)$, where $z\in R_q^m$ and $T$ is an arbitrary register.
    \item Sample $c\leftarrow \Btau$. 
    \item $y\leftarrow\calA_2(T,c)$, where $y\in R_q^{l}$.
    \item Output $1$ if $A y = z$, $\norminfty{y} \leq \gamma$, and $y_l = c$.
\end{enumerate}
\end{minipage}}
\end{center}
When $\tau, m, l, \gamma$ are functions of $\lambda\in \mathbb{N}$, we say that the $\MSIS$ hash function is chosen-coordinate binding (CCB) if for every $\poly$-time quantum algorithm $\calA$, $\Adv^{\CCB}_{\tau, m, l, \gamma}(\calA) \leq 1/\abs{B_\tau} + \negl$.
\end{definition}

\begin{definition}[Collapsing ($\Collapse$)]\label{def:collapse}
Let $m,l,\gamma \in \mathbb{N}$. The advantage of a quantum algorithm $\calA = (\calA_1,\calA_2)$ for winning $\Collapse_{m,l,\gamma}$, denoted $\Adv^{\Collapse}_{m,l,\gamma}$, is defined as $2p-1$ where $p$ is the probability the experiment below outputs $1$.
\begin{center}
\fbox{\begin{minipage}{0.8\textwidth}
\textbf{\textup{Experiment}} $\Collapse_{m,l,\gamma}$. 
\begin{enumerate}[topsep=3pt]
    \item Sample $A \leftarrow R_q^{m\times l}$.
    \item $(Y, Z, T) \leftarrow \calA_1(A)$, where $Y$ is a register on $R_q^l$, $Z$ is a register on $R_q^m$, and $T$ is an arbitrary register.
    \item Sample $b \leftarrow \{0,1\}$. If $b=1$, measure $Y$ in the computational basis.
    \item $b'\leftarrow \calA_2(Y,Z,T)$.
    \item Output $1$ if $b'=b$.
\end{enumerate}
\end{minipage}}
\end{center}
We say $\calA$ is valid if the state on the register $(Y,Z)$ output by $\calA_1$ in step 2 is supported on elements $(y,z)\in R_q^l\times R_q^m$ such that $A y = z$ and $\norminfty{y} \leq \gamma$. When $m,l,\gamma$ are functions of $\lambda \in \mathbb{N}$, we say that the $\MSIS$ hash function is collapsing if for every $\poly$-time quantum algorithm $\calA$, $\Adv^{\Collapse}_{m, l, \gamma}(\calA) \leq 1/2 + \negl$.
\end{definition}

\section{Security proof for \tSTMSIS}
\label{sec:security_stmsis}

The main result of this subsection is the following theorem which follows from \cref{prop:stmsis_nhstmsis,prop:nhstmsis_ccb,prop:ccb_collapsing,prop:collapsing_mlwe}. 

\begin{theorem}[$\STMSIS$ security]\label{thm:stmsis}
\nopagebreak
Let $m,k,\tau, \gamma,\eta \in \mathbb{N}$. Suppose $q\geq 16$, $q = 1 \modulo 2n$, and $2\gamma \eta n (m+k) < \floor{q/32}$. Suppose that there exists a quantum query algorithm $\calA$ for solving $\STMSIS_{H,\tau,m,k,\gamma}$ using $Q$ queries with expected advantage $\epsilon$ over uniformly random\footnote{Let $U \subset\{0,1\}^*$ be the query set of $\calA$, i.e., the finite subset of elements in $\{0,1\}^*$ that $\calA$ could possibly query (in particular $\abs{U} \leq 2^{O(\Time(\calA))})$. By uniformly random $H \colon \{0,1\}^* \to \Btau$, we mean that $H$ restricted to domain $U$ is uniformly random.} $H\colon \{0,1\}^* \to \Btau$. Then, for all $w \in \mathbb{N}$, there exists a quantum algorithm $\mathcal{B}$ that solves $\MLWE_{m+k,m,\eta}$ with advantage at least
\begin{align}\label{eq:stmsis}
    \frac{\epsilon-nq^{-k}}{4(2Q+1)^2}\Biggl(\frac{\epsilon-nq^{-k}}{(2Q+1)^2} - \frac{1}{\abs{\Btau}}\Biggr) -\frac{1}{4}\frac{1}{3^w}.
\end{align}
Moreover, $\Time(\calB) \leq  \Time(\calA) + \polytext(\log \abs{\Btau}, w,n,\log q,m,k)$.
\end{theorem}

Assuming that the choice of parameters as functions of the security parameter $\lambda$ is such that $nq^{-k} = \negl$, $1/\abs{\Btau} = \negl$, and $w = \poly$, \cref{thm:stmsis} shows that the advantage of $\calB$ is roughly $\Omega(\epsilon^2/Q^4)$.

The proof of \cref{thm:stmsis} proceeds by the following sequence of reductions, which we have labeled by the number of the section in which they are proven:
\begin{align*}
    \STMSIS \overset{\ref{sec:plainstmsis_to_stmsis}}{\longleftarrow} \nHSTMSIS \overset{\ref{sec:ccb_to_plainstmsis}}{\longleftarrow} \CCB \overset{\ref{sec:collapse_to_ccb}}{\longleftarrow} \Collapse \overset{\ref{sec:mlwe_to_collapse}}{\longleftarrow} \MLWE.
\end{align*}
First, we establish some properties of $R_q$ that will be used in \cref{sec:plainstmsis_to_stmsis,sec:mlwe_to_collapse}.

\subsection{Properties of \texorpdfstring{$R_q$}{Rq}}
\label{sec:ring_properties}

\begin{lemma}\label{lem:primitive_sum}
    Suppose $q = 1 \modulo 2n$. Let $w$ be a primitive $(2n)$-th root of unity in $\Z_q$. Then for all $m\in \Z$ such that $0\neq \abs{m}<n$, the following equation holds in $\Z_q$: $\sum_{j=0}^{n-1} w^{2mj} = 0$.
\end{lemma}
\begin{proof}
    Consider the following equation in $\Z_q$: $(1-w^{2m})\cdot  \sum_{j=0}^{n-1} w^{2mj} = 1-w^{2mn} = 0$, where the first equality uses a telescoping sum and the second uses $w^{2n}=1$. But $1-w^{2m}\neq 0$ since $0\neq \abs{m}<n$ and $w$ is a primitive $(2n)$-th root of unity in $\Z_q$. Therefore, since $\Z_q$ is an integral domain when $q$ is prime, $ \sum_{j=0}^{n-1} w^{2mj} = 0$ as required.
\end{proof}

\begin{lemma}\label{lem:ring_isomorphism}
    Suppose $q = 1 \modulo 2n$. Then, $R_q \cong \Z_q^n$ as algebras over $\Z_q$.\footnote{To be clear, the algebra $\Z_q^n$ over $\Z_q$ refers to the set $\Z_q^n$ equipped with component-wise addition and multiplication, and scalar multiplication defined by $\alpha\cdot (c_0,\dots,c_{n-1}) \coloneqq (\alpha c_0,\dots, \alpha c_{n-1})$, where $\alpha \in \Z_q$ and $(c_0,\dots,c_{n-1})\in \Z_q^n$.}
\end{lemma}
\begin{proof}
    For $q$ prime, the multiplicative group $\Z_q^*$ of non-zero elements in $\Z_q$ is cyclic. Let $g$ be a generator of $\Z_q^*$. Let $w\coloneqq g^{(q-1)/(2n)}$, which is well-defined since $q=1 \mod 2n$. Define the mapping $\phi\colon R_q \to \Z_q^n$ by:
    \begin{align}
        \phi(p(x)) = \begin{pmatrix}
            1 & w & \dots & w^{n-1}
            \\
            1 & w^3 & \dots & w^{3(n-1)}
            \\
            \vdots & &\ddots &\vdots 
            \vspace{2pt}
            \\
            1 & w^{(2n-1)} & \dots & w^{(2n-1)(n-1)}
        \end{pmatrix} \begin{pmatrix}
            a_0\\
            a_1\\
            \vdots\\
            a_{n-1}
        \end{pmatrix},
    \end{align}
    where $p(x) \coloneqq a_0 + a_1x + \dots + a_{n-1}x^{n-1}$. It is clear that $\phi$ is a linear map. To see that $\phi$ is homomorphic with respect to multiplication, observe that for any $\tilde{p}(x)\in \Z_q[x]$ such that $p(x) = \tilde{p}(x) \modulo (x^n+1)$, we have
\begin{align}
       \phi(p(x)) = \bigl(\tilde{p}\bigl(w^1\bigr),\tilde{p}\bigl(w^3\bigr),\dots, \tilde{p}\bigl(w^{(2n-1)}\bigr)\bigr),
\end{align}
since $(w^{2k-1})^n + 1 = 0$ in $\Z_q$ for all $k \in [n]$.

To see that $\phi$ is bijective, observe its explicit inverse $\phi'\colon \Z_q^n \to R_q$, defined by 
    \begin{align}
    \phi'(c_0,\dots,c_{n-1}) =& a_0 + a_1 x + \dots + a_{n-1}x^{n-1}, \quad \text{where}
    \\
        \begin{pmatrix}
            a_0\\
            a_1\\
            \vdots\\
            a_{n-1}
        \end{pmatrix} \coloneqq& n^{-1}\begin{pmatrix}
            1 & 1 & \dots & 1
            \\
            w^{-1} & w^{-3} & \dots & w^{-(2n-1)}
            \\
            \vdots & &\ddots &\vdots 
            \vspace{2pt}
            \\
            w^{-(n-1)} & w^{-3(n-1)} & \dots & w^{-(2n-1)(n-1)}
        \end{pmatrix} \begin{pmatrix}
            c_0\\
            c_1\\
            \vdots\\
            c_{n-1}
        \end{pmatrix}
    \end{align}
    and $n^{-1}$ denotes the multiplicative inverse of $n$ in $\mathbb{Z}_q$, which exists since $q = 1 \modulo  2n \implies n < q$. Since $w$ is a primitive $(2n)$-th root of unity in $\Z_q$, \cref{lem:primitive_sum} implies that the matrices corresponding to $\phi$ and $\phi'$ multiply to the identity in $\Z_q$. Therefore, $\phi'$ is the inverse of $\phi$.
\end{proof}

\subsection{Reduction from \tnHSTMSIS to \tSTMSIS}
\label{sec:plainstmsis_to_stmsis}

\begin{proposition}\label{prop:stmsis_nhstmsis}
Suppose $q = 1 \modulo 2n$. Let $m,k,\gamma,\tau \in \mathbb{N}$ and $H \colon \{0,1\}^* \to B_\tau$. Suppose that there exists a quantum query algorithm $\calA$ using $Q$ queries that solves $\STMSIS_{H,\tau,m,k,\gamma}$ with advantage $\epsilon$, then there exists a quantum query algorithm $\calB$ using $Q$ queries for solving $\nHSTMSIS_{H,\tau,m,m+k,\gamma}$ with advantage at least $\epsilon - n/q^k$. Moreover, $\Time(\calB) \leq \Time(\calA) + O(n\log (q) \cdot mk \min(m,k))$.
\end{proposition}

\begin{proof}
    The probability that a uniformly random $B \leftarrow \Z_q^{m\times (m+k)}$  has row-echelon form $[I_m|B']$ (i.e., rank $m$) is at least $(1-1/q^k)$. Therefore, by \cref{lem:ring_isomorphism}, the probability that a uniformly random $A \leftarrow R_q^{m\times (m+k)}$ does not have row-echelon form $[I_m|A']$ is at most $1-(1-1/q^k)^n \leq n/q^k$. When $A$ has row-echelon form $[I_m|A]$, $\calB$ first performs row reduction and then runs $\calA$. Since the time to perform row reduction on $A$ is $O(n\log (q) \cdot mk \min(m,k))$, the proposition follows.
\end{proof}

\subsection{Reduction from \tCCB to \tnHSTMSIS}
\label{sec:ccb_to_plainstmsis}

Let $S, U, C, R$ be alphabets, $V\colon S\times U\times C\times R \to \{0,1\}$, and $\calB = (\calB_1,\calB_2)$ be a quantum algorithm. We define the $\Sigma$-experiment by:
\begin{center}
\fbox{\begin{minipage}{0.8\textwidth}
\textbf{$\Sigma$-\textup{experiment}}. 
\begin{enumerate}[topsep=3pt]
    \item $s\leftarrow S$.
    \item $(u, T)\leftarrow \calB_1(s)$, where $u\in U$ and $T$ is an arbitrary register.
    \item $c \leftarrow C$.
    \item $r \leftarrow \calB_2(T,c)$.
    \item Output $1$ if $V(s,u,c,r) = 1$.
\end{enumerate}
\end{minipage}}
\end{center}
The advantage of $\calB$ for winning the $\Sigma$-experiment is the probability of the experiment outputting $1$.

In this subsection, we use the following theorem from \cite{Reprogram}. 

\begin{theorem}[Measure-and-reprogram~{\cite[Theorem 2]{Reprogram}}]\label{lem:measure_and_reprogram}  Let $\calA$ be a quantum query algorithm using $Q$ queries  that takes input $s\in S$ and outputs $u\in U$ and $r\in R$. Then, there exists a two-stage quantum algorithm $\calB = (\calB_1,\calB_2)$ (not using any queries) such that the advantage of $\calB$ in the $\Sigma$-experiment is at least
\begin{align}
    \frac{1}{(2Q+1)^2} \Pr\Bigl[ V(s,u,H(u),r) \Bigm| H \leftarrow C^U, \ s \leftarrow S, \ (u,r) \leftarrow \calA^{\ket{H}}(s) \Bigr].
\end{align}
Moreover, $\Time(\calB_1) + \Time(\calB_2) \leq \Time(\calA)$.
\end{theorem}

In the original statement of the theorem,  $\Time(\calB_1)+ \Time(\calB_2)$ is upper bounded by $\Time(\calA)-Q+\polytext(Q, \log(\abs{U}, \log(\abs{C}))$. (The $-Q$ is because our definition of $\Time(\calA)$ includes a $+Q$ term.) The term $\polytext(Q, \log(\abs{U}, \log(\abs{C}))$ accounts for the cost of instantiating $Q$ queries to a $2(Q+1)$-wise independent hash function family from $U$ to $C$. By the well-known Vandermonde matrix method (see, e.g., \cite[Section 6]{SecureEncryptionZhandry}), this cost can be upper bounded by $O(Q^2 \cdot \log(\abs{U})\cdot \log(\abs{C}))$. However, we follow the convention in \cite[Section 2.1]{Concrete} and equate this cost to $Q$ under the fair assumption that $\calB$, like $\calA$, can also query a random oracle at unit cost. 

\begin{proposition}\label{prop:nhstmsis_ccb}
Let $m,l,\gamma,\tau\in \mathbb{N}$. Suppose there exists a quantum query algorithm $\calA$ for solving \newline $\nHSTMSIS_{H,\tau,m,l,\gamma}$ using $Q$ queries with expected advantage $\epsilon$ over uniformly random $H\colon \{0,1\}^* \to \Btau$. Then there exists a quantum algorithm $\calB = (\calB_1,\calB_2)$ for winning $\CCB_{\tau,m,l,\gamma}$ with advantage at least $\epsilon/(2Q+1)^2$. Moreover $\Time(\calB_1) + \Time(\calB_2) \leq \Time(\calA)$.
\end{proposition}

\begin{proof}
The quantum query algorithm $\calA$ for $\nHSTMSIS_{H,\tau,m,l,\gamma}$ takes input $A$ and outputs $(y,M)$. So there exists another quantum query algorithm $\calA'$ using $Q$ queries that outputs
$((A y \concat M), \ y)$.

The first part of the proposition follows from applying \cref{lem:measure_and_reprogram} to $\calA'$ with the following parameter settings which make the $\Sigma$-experiment identical to the $\CCB_{\tau,m,l,\gamma}$ experiment:
\begin{enumerate}
    \item Set $S=R_q^{m\times l}$, $U$ to be the query set of $\calA'$, $C = \Btau$, and $R = R_q^l$.
    \item Set $V\colon R_q^{m\times l}\times U \times \Btau \times R_q^l \to \{0,1\}$ by
    \begin{align}
        V(A,u,c,y) = \1[z = A y, \ \norminfty{y}\leq \gamma, \ y_l = c],
    \end{align}
    where $u\in \{0,1\}^*$ is parsed as $u = (z \concat M)$ with $z \in R_q^m$ and $M \in \{0,1\}^*$.
\end{enumerate}
\end{proof}

\subsection{Reduction from \tCollapse to \tCCB}
\label{sec:collapse_to_ccb}

In this subsection, we will use the following lemma, which can be found as \cite[Proposition 29]{NecessityCollapsing}. 
\begin{lemma}\label{lem:two_projections}
    Let $P,Q$ be projectors in $\mathbb{C}^{d\times d}$ and $\rho$ be a density matrix in $\mathbb{C}^d$ such that $\rho Q = \rho$. Then $\tr(QP\rho P) \geq \tr(P\rho)^2$.
\end{lemma}

The following proposition is similar to \cite[Theorem 32]{CollapsingWithoutRoUnruh} and \cite[Theorem 28]{NecessityCollapsing} except the size of the challenge set in the $\CCB$ experiment (in step 3 of \cref{def:ccb}) is not restricted to being $2$.

\begin{proposition}\label{prop:ccb_collapsing}
Let $m,l,\gamma,\tau \in \mathbb{N}$. 
Suppose that there exists a quantum algorithm $\calA = (\calA_1,\calA_2)$ that succeeds in $\CCB_{\tau,m,l,\gamma}$ with advantage $\epsilon$, then there exists a valid quantum algorithm $\calB = (\calB_1,\calB_2)$ that succeeds in $\Collapse_{m,l,\gamma}$ with advantage at least
$\epsilon(\epsilon - 1/\abs{\Btau})$. Moreover, $\Time(\calB_1) \leq \Time(\calA_1) + \Time(\calA_2) + O(ml\log(q)\log(\abs{\Btau}))$ and $\Time(\calB_2) \allowbreak \leq \Time(\calA_2) + O(\log(\abs{\Btau}))$.
\end{proposition}

\begin{proof}
We assume without loss of generality (wlog) that the arbitrary register in step 2 of the $\CCB_{\tau,m,l,\gamma}$ experiment (\cref{def:ccb}) is of the form $(Y,T')$, where $Y$ is a register on $R_q^l$ and $T'$ is an arbitrary register. We assume wlog that $\calA_1$ prepares a state $\ket{\phi}$ on register $(Y,Z,T')$, where $Z$ is a register on $R_q^m$, and measures $Z$ in the computational basis to produce the $z$ in step 2 of the $\CCB_{\tau,m,l,\gamma}$ experiment. We also assume wlog that $\calA_2$ acts on its input register $(Y,T',C)$, where $C$ is a register on $\Btau$ that contains the $c$ from step 3 of the $\CCB_{\tau,m,l,\gamma}$ experiment, as follows:
\begin{enumerate}
    \item Apply a unitary $U$ of the form $\sum_{r\in \Btau} U_r\otimes \ketbrasame{r}$ on $(Y, T', C)$.
    \item Measure $Y$ in the computational basis.
\end{enumerate}

We proceed to construct $\calB = (\calB_1,\calB_2)$ for the $\Collapse_{m,l,\gamma}$ experiment (\cref{def:collapse}). We first construct $\calB_1$, given input $A\in R_q^{m\times l}$, as follows:
\begin{enumerate}
    \item Run $\calA_1(A)$ to prepare state $\ket{\phi}$ on register $(Y,Z,T')$.
    \item Prepare state $\ket{\psi}\coloneqq \abs{\Btau}^{-1/2}\sum_{r\in \Btau} \ket{r}$ on register $C$ in time $O(\log(\abs{\Btau}))$. The current state on register $(Y,Z,T',C)$ is $\sigma \coloneqq \ketbrasame{\phi} \otimes \ketbrasame{\psi}$. Apply $U$ on register $(Y,T',C)$ and then measure register $(Y, Z, T', C)$ with the projective measurement $\{\Pi, 1 - \Pi\}$, where $\Pi$ is defined by
    \begin{align}\label{proj_ccb_success}
        \Pi \coloneqq \sum_{r\in \Btau} \ \sum_{\substack{(y,z) \in R_q^l \times R_q^m \colon \\ \norminfty{y}\leq \gamma, \ Ay=z, \ y_l = r}} \ketbrasame{y,z} \otimes 1_{T'} \otimes \ketbrasame{r}.
    \end{align}
    This measurement can be implemented by computing a bit indicating whether the constraints defining $\Pi$ are satisfied into a separate register and then measuring that register, which takes time $O(ml\log(q) + \log(\abs{\Btau}))$.
    \item Let $B$ be a bit register. If $\Pi$ is measured, set the bit stored in $B$ to $1$. If $(1-\Pi)$ is measured, replace the state on register $(Y,Z)$ with $\ket{0^l}\otimes\ket{0^m}$, set the bit stored in $B$ to $0$.  Then output the register $(Y, Z, T', C, B)$.
\end{enumerate}

Let $T \coloneqq (T', C, B)$. We construct $\calB_2$, given input register $(Y,Z,T)$:
\begin{enumerate}
    \item If $B$ contains $0$, output a uniformly random bit $b' \in \{0,1\}$.
    \item Else apply $U^\dagger$ on register $(Y, T', C)$. Then measure $C$ with the projective measurement $\{\ketbrasame{\psi}, 1 - \ketbrasame{\psi}\}$ using (the inverse of) the preparation circuit for $\ket{\psi}$ in time $O(\log(\abs{\Btau})$. If the outcome is $\ketbrasame{\psi}$, output $0$; else output $1$.
\end{enumerate}

It is clear that $\calB$ is valid by definition. Moreover,
\begin{align}
    \Time(\calB_1) &\leq \Time(\calA_1) + \Time(\calA_2) + O(ml\log(q)\log(\abs{\Btau})),
    \\
    \Time(\calB_2) &\leq \Time(\calA_2) + O(\log(\abs{\Btau})).
\end{align}

We proceed to lower bound the success probability of $\calB$. We analyze the probabilities of the following disjoint cases corresponding to $\calB$ being successful. 
\begin{enumerate}
    \item Case 1: In this case, $1-\Pi$ is measured and $b'=b$. The probability that $1-\Pi$ is measured is $(1-\epsilon)$. Conditioned on $1-\Pi$ being measured, $b'$ is a uniformly random bit so the probability $b'=b$ is $1/2$. Therefore, the overall probability of this case is $(1-\epsilon)/2$.
    \item Case 2: In this case, $\Pi$ is measured, $b=1$, and then $1-\ketbrasame{\psi}$ is measured.  The probability that $\Pi$ is measured is $\epsilon$ and the probability that $b=1$ is $1/2$. We now condition on these two events happening. Since $b=1$, the state of register $C$ in the input to $\calB_2$ is a mixture of states of the form $\ketbrasame{r}$ where $r\in \Btau$. This is because $b=1$ means that register $Y$ is measured in the computational basis and conditioned on $\Pi$ being measured, the $C$ register is also measured in the computational basis (see the form of $\Pi$ in \cref{proj_ccb_success}). Therefore, the probability of $\calB_2$ measuring $\ketbrasame{\psi}$ is $1/\abs{\Btau}$. Therefore, the overall probability of this case is $\epsilon \cdot (1/2) \cdot (1- 1/\abs{\Btau})$.
    \item Case 3: In this case, $\Pi$ is measured, $b=0$, and then $\ketbrasame{\psi}$ is measured.  The probability that $b=0$ is $1/2$. Conditioned on $b=0$,  \cref{lem:two_projections}, applied with projectors $\ketbrasame{\psi}$ and $U^\dagger \Pi U$ and state $\sigma$, shows that the probability of measuring $\Pi$ and then $\ketbrasame{\psi}$ is least $\epsilon^2$. Therefore, the overall probability of this case is at least $\epsilon^2/2$.
\end{enumerate}

Summing up the probabilities of the above cases, we see that the success probability of $\calB$ is at least
\begin{align}
  \frac{1-\epsilon}{2} + \frac{\epsilon}{2} \Bigl(1 - \frac{1}{\abs{\Btau}}\Bigr) + \frac{\epsilon^2}{2} = \frac{1}{2} + \frac{\epsilon}{2}\Bigl(\epsilon - \frac{1}{\abs{\Btau}}\Bigr).
\end{align}
Therefore, the advantage of $\calB$ is at least $\epsilon(\epsilon - 1/\abs{\Btau})$, as required.
\end{proof}

\subsection{Reduction from \tMLWE to \tCollapse}
\label{sec:mlwe_to_collapse}

The proof structure of the main result of this subsection, \cref{prop:collapsing_mlwe}, follows \cite[Theorem 1]{QuantumMoneyLiuMontgomeryZhandry}. We need to modify a number of aspects of their proof since it applies to the $\textsf{SIS}$ hash function whereas here we consider its module variant, i.e., the $\MSIS$ hash function.

We will use a rounding function $\roundt{\cdot}\colon \Z_q \to \{0,1,\dots,t-1\}$, where $t\in \mathbb{N}$, that is defined as follows. For $j \in \{0,1,\dots, t-1\}$, define
\begin{align}
I_j \coloneqq 
\begin{cases}
    \{j\floor{q/t},j\floor{q/t}+1,\dots, j\floor{q/t} + \floor{q/t}-1\} &\text{if } j \in \{0,1,\dots, t-2\}, \\
    \{(t-1)\floor{q/t},(t-1)\floor{q/t}+1,\dots, q-1\} &\text{if } j=t-1.
\end{cases}
\end{align}
(Note that $I_j$ contains exactly $\floor{q/t}$ elements for $j\in \{0,1,\dots,t-2\}$ and at least $\floor{q/t}$ elements for $j=t-1$ with the constraint that $q/t \leq \abs{I_{t-1}} \leq q/t + t - 1$.) Then, for $a\in \Z_q$, define $\roundt{a}$ to be the unique $j\in \{0,1,\dots,t-1\}$ such that $a\in I_{j}$. 

We will also use the following convenient notation. Let $Y$ and $Z$ be registers and $f: Y \to Z$. The measurement $y\mapsto f(y)$ on register $Y$ refers to the measurement implemented by computing $f(y)$ into a separate register $Z$, measuring $Z$ in the computational basis, and discarding the result.

Finally, we will use the following lemma.
\begin{lemma}\label{lem:uniform_distribution}
    Let $0\neq \Delta \in R_q^l$ and $\alpha\in \{0,\dots,n-1\}$. If $b \leftarrow R_q^l$,  then $(b\cdot \Delta)_\alpha$ is uniformly distributed in $\Z_q$.
\end{lemma}

\begin{proof}
    Writing $b = (b_1,\dots,b_l)$ and $\Delta = (\Delta_1,\dots,\Delta_l)$, we have 
    \begin{align}
        (b \cdot \Delta)_\alpha = (b_1 \Delta_1)_\alpha + \dots + (b_l \Delta_l)_\alpha.
    \end{align}
    Since $\Delta\neq 0$, there exists an $i\in [l]$ such that $\Delta_i\neq0$. To prove the lemma, it suffices to prove that $(b_i\Delta_i)_a$ is uniformly distributed in $\Z_q$.

    Let $\phi, \phi'$ be as defined in the proof of \cref{lem:ring_isomorphism}. Write $\phi(\Delta_i) = (c_0,\dots,c_{n-1})\in \Z_q^n$. Since $\Delta_i\neq 0$ there exists $j\in \{0,\dots,n-1\}$ such that $c_j \neq 0$. Since $b_i$ is a uniformly random element of $R_q$, $\phi(b_i)$ is a uniformly random element of $\Z_q^n$. Therefore, the distribution of $(b_i\Delta_i)_\alpha = \phi'(\phi(b_i)\phi(\Delta_i))_\alpha$ (where we used \cref{lem:ring_isomorphism} for the equality) is the same as the distribution of
    \begin{align}
        \phi'(d_0 c_0, \dots ,d_{n-1}c_{n-1})_\alpha, \quad \text{where} \quad d_0,\dots,d_{n-1}\leftarrow \Z_q.
    \end{align}
    By the linearity of $\phi'$,
        \begin{align}
        \phi'(d_0 c_0, \dots ,d_{n-1}c_{n-1})_\alpha = d_j c_j\phi'(e_j)_\alpha + \sum_{j'\neq j} d_{j'}c_{j'} \phi'(e_{j'})_\alpha,
    \end{align}
    where $e_j$ denotes the $j$th standard basis vector of $\Z_q$. 
    
    But $\phi'(e_j)_\alpha = n^{-1} \cdot w^{-(2j+1)\alpha} \neq 0$ (see \cref{lem:ring_isomorphism}). Therefore $d_j c_j\phi'(e_j)_\alpha$ is uniformly distributed in $\Z_q$ if $d_j\leftarrow \Z_q$. Hence $(b_i\Delta_i)_\alpha$ is uniformly distributed in $\Z_q$ as required.
\end{proof}

The main result of this subsection is the following proposition.
\begin{proposition}\label{prop:collapsing_mlwe}
    Let $m,l,\gamma,\eta \in \mathbb{N}$. Suppose $q\geq 16$ and $2\gamma \eta n l < \floor{q/32}$. Suppose there exists a quantum algorithm $\calA$ that succeeds in $\Collapse_{m,l,\gamma}$ with advantage $\epsilon$. Then, for all $w\in \mathbb{N}$, there exists a quantum algorithm $\calB$ that solves $\MLWE_{l,m,\eta}$ with advantage at least $(\epsilon - 3^{-w})/4$. Moreover, $\Time(\calB) \leq \Time(\calA) + \polytext(w)$.
\end{proposition}

Before proving this proposition, we first prove two lemmas. Let $Y$ be a register on $R_q^l$ and $A \in R_q^{m\times l}$. For $t \in \mathbb{N}$, we define the following measurements on $Y$:
\begin{itemize}
    \item $M_0$: computational basis measurement.
    \item $M_1^t$: sample $e_1 \leftarrow S_\eta^m$, $e_2\leftarrow S_\eta^l$, set $b \coloneqq e_1^\top A  + e_2^\top \in R_q^l$, sample $s \leftarrow R_q$, then perform measurement $y \mapsto \roundt{(b\cdot y + s)_0}$.
    \item $M_2^t$: sample $b \leftarrow R_q^l$, $s\leftarrow R_q$, then perform measurement $y \mapsto \roundt{(b\cdot y + s)_0}$.
\end{itemize}

\begin{lemma}
\label{lem:measurement_lwe}
Let $t\in \mathbb{N}$ be such that $2\gamma \eta nl < \floor{q/t}$. For all $y,y' \in R_q^l$ with $Ay = Ay'$ and $\norminfty{y'},\norminfty{y} \leq \gamma$,
\begin{align}
    M_1^t(\ketbra{y}{y'}) = \Bigl(1 - \frac{t}{q}\cdot \expect\Bigl[\abs{e \cdot(y - y')}_0 \Bigm| e \leftarrow S_\eta^l\Bigr]\Bigr)\ketbra{y}{y'}.
\end{align}
\end{lemma}

\begin{proof}
We have
\begin{equation}
     M_1^t(\ketbra{y}{y'}) = \Pr\Big[\roundt{(b\cdot y + s)_0} = \roundt{(b\cdot y' + s)_0} \Bigm| e_1 \leftarrow S_\eta^m, e_2\leftarrow S_\eta^l, b \coloneqq e_1^\top A  + e_2^\top, s\leftarrow R_q\Big] \cdot \ketbra{y}{y'}.
\end{equation}
Writing $z\coloneqq Ay = Ay'$, we have
\begin{align}
    b\cdot y + s = (e_1\cdot z + s) + e_2\cdot y \quad \text{and} \quad b\cdot y' + s = (e_1\cdot z + s) + e_2\cdot y'.
\end{align}
The result follows by observing that $\abs{e_2\cdot(y-y')}_0 \leq \norminfty{e_2} \cdot \norminfty{y-y'} \cdot nl \leq 2\gamma \eta nl < \floor{q/t}$ and $(e_1\cdot z + s)$ is a uniformly random element of $R_q$.
\end{proof}

\begin{lemma}
\label{lem:measurement_plain}
Let $t\in \mathbb{N}$ be such that $t^2\leq q$. Then there exists $0\leq p_t \leq 2/t$ such that for all $y,y' \in R_q^l$ with $y'\neq y$, we have 
\begin{align}
M_2^t(\ketbrasame{y}) =\ketbrasame{y} \quad \text{and} \quad M_2^t(\ketbra{y}{y'}) = p_t \ketbra{y}{y'}.
\end{align}
\end{lemma}

\begin{proof}
The first equality is clearly true. For the second, observe that
    \begin{align}
        M_2^t(\ketbra{y}{y'}) = \Pr[\roundt{(b\cdot y + s)_0} = \roundt{(b\cdot y' + s)_0} \mid b \leftarrow R_q^l, s \leftarrow R_q]. 
    \end{align}
    Write $y' = y + \Delta$ for some $0\neq \Delta \in R_q^l$. Then,  $(b\cdot \Delta)_0$ is uniformly distributed in $\mathbb{Z}_q$ by \cref{lem:uniform_distribution}. Therefore, writing $p_t \coloneqq \Pr[ \roundt{u} = \roundt{u + v} \mid u,v \leftarrow \Z_q]$, we have
    \begin{equation}
    \begin{aligned}
    \Pr[\roundt{(b\cdot y + s)_0} =& \roundt{(b\cdot y' + s)_0} \mid b \leftarrow R_q^l, s \leftarrow R_q]
    \\ =& p_t = 1 - \Bigl(\frac{(t-1)\floor{q/t}}{q} \cdot \frac{q-\floor{q/t}}{q} + \frac{\abs{I_{t-1}}}{q} \cdot \frac{q - \abs{I_{t-1}}}{q}\Bigr) \leq \frac{1}{t} + \frac{t}{q} \leq \frac{2}{t},
    \end{aligned}
    \end{equation}
    where the last inequality uses $t^2\leq q$.
\end{proof}

Combining \cref{lem:measurement_lwe,lem:measurement_plain} gives the following corollary.

\begin{corollary}\label{cor:measurements}
    Let $t,d\in \mathbb{N}$ be such that $2\gamma \eta n l < \floor{q/(td)}$ and $t^2\leq q$. Let $\rho$ be a density matrix on register $Y$. Suppose there exists $z\in R_q^m$ such that $\rho$ is supported on $\{y \in R_q^l \mid Ay = z, \norminfty{y} \leq \gamma\}$. Then
    \begin{align}
        &M_1^t(\rho) = \frac{1}{d} M_1^{t d}(\rho) + \Bigl(1- \frac{1}{d}\Bigr)\rho, \\
        &M_2^t(\rho) = \frac{1}{d}M_0(M_1^{t d}(\rho)) + \Bigl(1- \frac{1}{d} - p_t\Bigr) M_0(\rho) + p_t \rho,
    \end{align}
    where $p_t$ is as defined in \cref{lem:measurement_plain}.
\end{corollary}
\begin{proof}
    The first equality is immediate. The second equality follows from the observation that $M_0(M_1^{td}(\rho)) = M_1^{td}(M_0(\rho))$ since $M_0$ and $M_1$ both act on $\rho$ by entry-wise multiplication.
\end{proof}

Given the above lemmas, \cref{prop:collapsing_mlwe} follows from the proof of \cite[Theorem 1]{QuantumMoneyLiuMontgomeryZhandry}. The high-level idea of the proof is that $M_1^t$ is close to the identity operation while $M_2^t$ is close to $M_0$. Therefore, if the identity operation can be efficiently distinguished from $M_0$, then $M_1^t$ and $M_2^t$ can be efficiently distinguished, which solves the $\MLWE$ problem.  For completeness, we give the details below.
\begin{proof}[Proof of \cref{prop:collapsing_mlwe}]
Let $t \coloneqq 4$ and $d \coloneqq 8$ so that $g\coloneqq 1 - 1/d - p_t \geq 3/8$ and $dg \geq 3$, where $p_t$ is as defined in \cref{lem:measurement_plain}. Let $\calA = (\calA_1,\calA_2)$ be a valid algorithm for the $\Collapse_{m,l,\gamma}$ experiment (\cref{def:collapse}) with advantage $\epsilon$. 

Fix $w \in \mathbb{N}$ and $A \in R_q^{m\times l}$. Let $T\coloneqq \sum_{j=0}^{w - 1} (dg)^{-j}$ and let $\calB$ be the quantum algorithm defined on input $b \in R_q^l$ as follows:
\begin{enumerate}
    \item Create state $\rho$ on register $(Y, Z, T)$ by running $\calA_1(A)$.
    \item Sample $j\in \{0,1,\dots, w-1\}$ with probability $(dg)^{-j}/T$. 
    \item Apply $M_1^{t d}$ to $\rho$ on the $Y$ register for $j$ times. Call the resulting state $\rho_j$.
    \item Sample $s\leftarrow R_q$ and apply the measurement $x\mapsto \roundt{(b\cdot x + s)_0}$ to $\rho_j$ on the $Y$ register to give state $\rho_j'$.
    \item Compute bit $b'\in \{0,1\}$ by running $\calA_2(\rho_j')$.
    \item Output $b'$ if $j$ is even and $1-b'$ if $j$ is odd.
\end{enumerate}
For $j\in \{0,1,\dots, w-1\}$, let $\epsilon_j$ denote the signed distinguishing advantage of $\calA_2$ on inputs $\rho_j$ versus $M_0(\rho_j)$, i.e., $\epsilon_j \coloneqq \Pr[\calA_2(\rho_j) = 0] - \Pr[\calA_2(M_0(\rho_j)) = 0]$, and let $\delta_j$ denote the signed distinguishing advantage of $\calA_2$ on inputs $M_1^t(\rho_j)$ versus $M_2^t(\rho_j)$. Then the signed distinguishing advantage of $\calB$ on input distributions $\left[e_1\leftarrow S_\eta^m, \ e_2\leftarrow S_\eta^l, \ b\coloneqq e_1^\top A + e_2^\top\right]$ versus $\left[b^\top\leftarrow R_q^l\right]$ is 
\begin{align}
    \delta \coloneqq \frac{1}{T} \sum_{j=0}^{w -1} (-dg)^{-j} \delta_{j},
\end{align}
because $\rho_j' = M_1^t(\rho_j)$ if $b$ is sampled according to $\left[e_1\leftarrow S_\eta^m, \ e_2\leftarrow S_\eta^l, \ b\coloneqq e_1^\top A\right.$ $\left.+ e_2^\top\right]$ and $\rho_j' = M_2^t(\rho_j)$ if $b$ is sampled according to $\left[b^\top\leftarrow R_q^l\right]$.

By \cref{cor:measurements} (which applies by the assumptions in the proposition and the validity of $\calA$), we have $\delta_j = \frac{1}{d} \epsilon_{j+1} + g\epsilon_j$ for all $j\in \{0,1,\dots, w-2\}$. Therefore,
\begin{align}\label{eq:distinguishing_relation}
    \epsilon_i(-dg)^{-i} = \epsilon_0 - \frac{1}{g}\sum_{j=0}^{i-1}(-dg)^{-j} \delta_j \quad \text{for all $i\in \{0,1,\dots, w-1\}$}.
\end{align}
Then,
\begin{align}\label{reduction_fixed_A}
    \delta = \frac{g}{T}(\epsilon_0 - \epsilon_w(-dg)^{-w}).
\end{align}

We now unfix $A\in R_q^{m\times l}$ and take the expectation of \cref{reduction_fixed_A} over $A \leftarrow R_q^{m\times l}$ to see that
\begin{align}
    \abs{\expect_A[\delta]} = 
    \frac{g}{T}\Bigl|\expect_A[\epsilon_0 - \epsilon_w(-dg)^{-w}]\Bigr| \geq \Bigl(g-\frac{1}{d}\Bigr)(\epsilon - (dg)^{-w}) \geq \frac{1}{4}\Bigl(\epsilon - \frac{1}{3^w}\Bigr),
\end{align}
where the first inequality uses $T\leq dg/(dg-1)$, $\abs{\epsilon_w}\leq 1$, and $\epsilon = \abs{\expect_A[\epsilon_0]}$. 

Since $\Time(\calB) = \Time(\calA) + \polytext(w)$ and $\abs{\expect_A[\delta]}$ is the advantage of $\calB$ for solving $\MLWE_{l,m,\eta}$, the proposition follows.
\end{proof}

\section{Concrete parameters}
\label{sec:concrete}

In this section, we describe how to adjust the parameter settings of $\Dilithium$ using \cref{thm:stmsis} to achieve security levels comparable to those considered in the $\Dilithium$  specifications \cite{Dilithium}, $\DilithiumQROM$ \cite{Concrete}, and the relevant NIST Federal Information Processing Standards (FIPS) \cite[Appendix A]{FIPS204}. 

We will use the same notation as in the $\Dilithium$ specification, \cite{Dilithium}. \cite{Dilithium} specifies $\Dilithium$ in terms of the variables
\begin{equation}
     q, n, k, l, H, \tau, d, \tau, \gamma_1, \gamma_2, \eta, \beta.
\end{equation}
Except for the variable $d$, these variables roughly specify $\Dilithium$ according to the simplified version given in \cref{fig:dilithium_simplified}. The variable $d$ specifies a further compression of the public key. To see how these variables precisely specify the full version of $\Dilithium$, we refer the reader to \cite{Dilithium}.

The security analysis of $\CRYSTALSDilithium$ in \cite{Concrete} leads to \cite[Eqs. (10) and (11)]{Concrete} which shows the following. Given a quantum query algorithm $\calA$ for breaking the $\sUFCMA$-security of $\Dilithium$, there exist quantum algorithms $\calB, \mathcal{D}, \mathcal{E}$ and quantum query algorithm $\calC$ such that $\Time(\calB) = \Time(\calC) =  \Time(\calA)$ and  $\Time(\calD) \approx \Time(\calA)$ with 
\begin{equation}\label{DilSecure}
    \Adv^{\sUFCMA}_{\Dilithium}(\calA) \le 2^{-\alpha+1} + \Adv_{k, l, \eta}^{\MLWE}(\calB) + \Adv_{H, \tau, k, l+1, \zeta}^{\STMSIS}(\calC) + \Adv_{k, l, \zeta'}^{\MSIS}(\mathcal{D}) + \Adv_{\Sam}^\PR(\calE)
\end{equation}
where $\zeta, \zeta'$ are functions of parameters $\gamma_1, \gamma_2, \beta, d, \tau$ defined as follows:
\begin{align}
    \zeta \coloneqq \max(\gamma_1 - \beta, 2\gamma_2+1+2^{d-1}\tau) \quad \text{and} \quad
    \zeta' \coloneqq \max(2(\gamma_1 - \beta), 4\gamma_2 + 2).
\end{align}
$\Adv_{\Sam}^\PR(\calE)$ is the advantage of any algorithm distinguishing between the pseudorandom function used by $\Dilithium$ and a randomly selected function; and $\alpha$ is a min-entropy term that can be bounded using \cite[Lemma C.1 of ePrint version]{Concrete} by
\begin{equation}
    \alpha \geq \min\Biggl(-n\log\Biggl(\frac{2\gamma_1 + 1}{2\gamma_2-1}\Biggr),-kl\log(n/q)\Biggr).
\end{equation}
In the QROM, we can construct an optimal pseudorandom function using a random oracle such that $\Adv_{\Sam}^\PR(\calE)$ is asymptotically negligible and can be neglected.

\cref{thm:stmsis} shows that the hardness of $\STMSIS$ in the QROM is at least that of $\MLWE$. Therefore, \cref{thm:stmsis} and \cref{DilSecure} rigorously imply the asymptotic result that, under suitable choices of parameters as functions of the security parameter $\lambda$, if there are no $\poly$-time quantum algorithms that solve $\MLWE$ or $\MSIS$ then there is no $\poly$-time quantum algorithm that breaks the $\sUFCMA$ security of $\Dilithium$. This is a very positive sign for the security of $\Dilithium$ as $\MSIS$ and $\MLWE$ are far better-studied problems and there is substantial support for the assumption that they are hard problems.

We proceed to give concrete estimates of the Core-SVP security of $\Dilithium$ under several choices of parameters using \cref{thm:stmsis} and \cref{DilSecure}. These estimates rely on some heuristic assumptions that we will clearly state. We remark that the concrete security estimates appearing in \cite{Concrete,Dilithium} use similar heuristic assumptions. 

We begin by dividing both sides of \cref{DilSecure} by $\Time(\calA)$. Using $\Time(\calB) = \Time(\calC) = \Time(\calA)$, assuming the approximation in $\Time(\calD) \approx \Time(\calA)$ can be replaced by equality, and using our work's parameters in \cref{tab:compare_dilithium_specs,tab:compare_dilithium_qrom,tab:noComp} for which $\alpha \geq 257$, we obtain
\begin{equation}
    \sfS(\calA) \leq 2^{-256} + \frac{\Adv_{k, l, \eta}^{\MLWE}(\calB)}{\Time(\calB)} + \frac{\Adv_{H, \tau, k, l+1, \zeta}^{\STMSIS}(\calC)}{\Time(\calC)} + \frac{\Adv_{k, l, \zeta'}^{\MSIS}(\mathcal{D})}{\Time(\calD)},
\end{equation}
where $\sfS(\calA) \coloneqq  \Adv^{\sUFCMA}_{\Dilithium}(\calA)/\Time(\calA)$.

By \cref{thm:stmsis}, for any $\eta' \in \mathbb{N}$ with $\eta' < \floor{q/32}/(2\zeta n(k+l+1))$, there exists a quantum algorithm $\calC'$ for $\MLWE_{k+l+1,k,\eta'}$, such that
\begin{equation}
    \sfS(\calA) \leq 2^{-256} + \frac{\Adv_{k, l, \eta}^{\MLWE}(\calB)}{\Time(\calB)} + \frac{8Q^2\sqrt{\Adv_{k+l+1, k, \eta'}^{\MLWE}(\calC')}}{\Time(\calC)} + \frac{\Adv_{k, l, \zeta'}^{\MSIS}(\mathcal{D})}{\Time(\calD)},
\end{equation}
where $Q$ is the number of queries $\calC$ uses and we assume that \cref{eq:stmsis} is well-approximated by $\epsilon^2/(64Q^4)$, in particular, that $\tau$ is sufficiently large.

Also by \cref{thm:stmsis}, we have $\Time(\calC')$ is at most $\Time(\calC)$ plus polynomial terms. Heuristically assuming that we can neglect the polynomial terms and using $Q \leq \Time(\calC)$, we obtain
\begin{equation}\label{eq:security_after_reduction}
   \sfS(\calA) \leq 2^{-256} + \frac{\Adv_{k, l, \eta}^{\MLWE}(\calB)}{\Time(\calB)} + 8Q^{3/2} \sqrt{\frac{\Adv_{k+l+1, k, \eta'}^{\MLWE}(\calC')}{\Time(\calC')}} + \frac{\Adv_{k, l, \zeta'}^{\MSIS}(\mathcal{D})}{\Time(\calD)}.
\end{equation}

Now, for NIST security level $l\in [5]$, we upper bound $Q$ by $B_l$, where $B_l$ is given in \cref{tab:nist_level_query}. 
\begin{table}[ht]
\centering\renewcommand{\arraystretch}{1.5}
    \begin{tabular}{c|p{1cm}|p{1cm}|p{1cm}|p{1cm}|p{1cm}}
        NIST Security Level (SL$l$) &  SL1 & SL2 & SL3 & SL4 & SL5\\
        \hline
        Upper bound on $Q$ $(B_l)$ &  $2^{64}$ & $2^{86}$ & $2^{96}$ & $2^{128}$ & $2^{128}$
    \end{tabular}
    \caption{Upper bounds on $Q$ for NIST security levels 1 to 5. These numbers are based on \cite[Appendix A]{FIPS204} together with well-known quantum query complexity results if we model the block ciphers and hash functions used in \cite[Appendix A]{FIPS204} as random functions.\protect\footnotemark}
    \label{tab:nist_level_query}
\end{table}

From the third term on the right-hand side of \cref{eq:security_after_reduction}, we see that the Quantum Core-SVP security of $\STMSIS$ can be estimated by 
\begin{equation}\label{eq:reduction_loss}
    \frac{z}{2} - \frac{3}{2} \log(B_l) - 3,
\end{equation}
where $z$ is the Quantum Core-SVP security of the associated $\MLWE$ problem.

Having reduced the $\sEUFNMA$ security of $\Dilithium$ to the security of standard lattice problems $\MLWE$ and $\MSIS$, we proceed to estimate their security. Following the analysis in the $\Dilithium$ specifications \cite{Dilithium}, we perform our security estimates via the Core-SVP methodology introduced in \cite{AlkimNewHope2016}. In the Core-SVP methodology, we consider attacks using the Block Khorkine-Zolotarev (BKZ) algorithm \cite{BkzOriginal1994,ChenNguyen2011}. The BKZ algorithm with block size $\mu\in \mathbb{N}$ works by making a small number of calls to an SVP solver on $\mu$-dimensional lattices. The Core-SVP methodology conservatively assumes that the run-time of the BKZ algorithm is equal to the cost of a single run of the SVP solver at its core. The latter cost is then estimated as $2^{0.265\mu}$ since this is the cost of the best quantum SVP solver \cite[Section C.1]{Dilithium} due to Laarhoven \cite[Section 14.2.10]{LaarhovenThesis2016}. Therefore, to estimate the security of an $\MLWE$ or $\MSIS$ problem, it suffices to estimate the smallest $\mu \in \mathbb{N}$ such that BKZ with block-size $\mu$ can solve the problem.  Then we say $0.265\mu$ is the \emph{Quantum Core-SVP} security of the problem. 

To describe how the block-size can be estimated, it is convenient to define the function $\delta\colon \mathbb{N} \to \mathbb{R}$,
\begin{equation}
   \delta(\mu) \coloneqq \Bigl(\frac{(\mu\pi)^{1/\mu} \mu}{2\pi e}\Bigr)^{\frac{1}{2(\mu-1)}}.
\end{equation}

\subsection{Concrete security of MLWE}
\label{sec:MLWE}

Our security analysis of $\MLWE$ generally follows the $\Dilithium$ specifications, \cite[Appendix C.2]{Dilithium}. For $a,b,\epsilon \in \mathbb{N}$, we first follow \cite[Appendix C.2]{Dilithium} and assume that $\MLWE_{a,b,\epsilon}$ is as hard as the Learning With Errors problem $\LWE_{na,nb,\epsilon}$. Then, for $a',b'
\in \mathbb{N}$, $\LWE_{a',b',\epsilon}$ is defined to be the same as $\MLWE_{a',b',\epsilon}$ with $n$ set to $1$ so that $R_q = \mathbb{Z}_q$.

Then, as done in \cite[Appendix C.2]{Dilithium}, we follow the security analysis in \cite{AlkimNewHope2016}. \cite{AlkimNewHope2016} considers two attacks based on the BKZ algorithm, known as the primal attack and dual attack. The block-size is then taken as the minimum of the block-sizes for the primal and dual attacks. These attacks are analyzed as follows.

\footnotetext{Given a random function $f\colon [N]\to [N]$, the number of quantum queries to $f$ needed to find a preimage of $1$ is $\Theta(N^{1/2})$ \cite{grover1996fast} and the number of quantum queries to $f$ needed to find a collision, i.e., $i\neq j$ such that $f(i)=f(j)$, is $\Theta(N^{1/3})$ \cite{NoteSetEqualityZhandry2015}. (We ignore the constants hidden in the $\Theta$-notation; more detailed analysis is possible, see, e.g., \cite{GroverDetailed2020}.)}

\begin{enumerate}
    \item Primal attack \cite[Section 6.3]{AlkimNewHope2016}. Let $c \coloneqq na + nb + 1$. Then, to solve $\LWE_{na,nb,\epsilon}$, we set the BKZ block-size $\mu$ to be equal to the smallest integer $\geq 50$ such that:\footnote{In \cite[Section 6.3]{AlkimNewHope2016}, the exponent on $\delta(\mu)$ is given as $2\mu-c-1$, but it was later corrected to $2\mu-c$ by \cite[Section 3.2]{AlbrechtExpectedCost2017}. There can be spurious solutions with $0<\mu<50$ for which the approximations leading to the inequality break down.} $\xi \sqrt{\mu} \leq \delta(\mu)^{2\mu-c} \cdot q^{na/c}$.

    \item Dual attack \cite[Section 6.4]{AlkimNewHope2016}. Let $c' \coloneqq na+nb$. Then to solve $\LWE_{na,nb,\epsilon}$, we set the BKZ block-size $\mu$ to be equal to the smallest integer $\geq 50$ such that $-2 \pi^2 \tau(\mu)^2 \geq \ln(2^{-0.2075\mu/2})$,
    where $\tau(\mu) \coloneqq \delta(\mu)^{c'-1}q^{nb/c'} \epsilon /q$.
\end{enumerate}

\subsection{Concrete security of MSIS}
\label{sec:MSIS}

Our security analysis of $\MSIS$ uses heuristics in the $\Dilithium$ specifications \cite[Appendix C.3]{Dilithium} and \cite{LyubashevskyLatticeSignatures2012} (which is in turn based on \cite{MicciancioRegev2009}).\footnote{We were unable to completely reuse the analysis in \cite[Section C.3]{Dilithium} as it is not completely described. Comparing the estimates for $\mu$ obtained by the method here with that in \cite[Table 1]{Dilithium}, we find our estimates are consistently around $4/5$ times that given in \cite[Table 1]{Dilithium}. Therefore, our estimates underestimate the security of $\MSIS$ compared to \cite{Dilithium}.} For $a,b,\xi \in \mathbb{N}$, we first follow \cite[Appendix C.3]{Dilithium} and assume that $\MSIS_{a,b,\xi}$ is as hard as the Short Integer Solutions problem $\SIS_{na,nb,\xi}$. Then, for $a',b'\in \mathbb{N}$, $\SIS_{a',b',\xi}$ is defined to be the same as $\MSIS_{a',b',\xi}$ with $n$ set to $1$ so that $R_q = \mathbb{Z}_q$. Following \cite{LyubashevskyLatticeSignatures2012}, we estimate the security of $\SIS_{na,nb,\xi}$, by considering the attack that uses the BKZ algorithm with block-size $\mu$ to find a short non-zero vector in the lattice
\begin{equation}
    L(A) \coloneqq \{y \in \mathbb{Z}^{na+nb} \mid [I_{na}\mid A] \cdot y = 0 \modulo q \},
\end{equation}
where $A \leftarrow \mathbb{Z}_q^{na\times nb}$. Following \cite[Eq.~(3) of ePrint version]{LyubashevskyLatticeSignatures2012}, the BKZ algorithm is expected to find a vector $v \in L(A)$ of Euclidean length\footnote{Compared to \cite[Eq.~(3) of ePrint version]{LyubashevskyLatticeSignatures2012}, we do not take the $\min$ of \cref{eq:sis_euclidean_length} with $q$ since ``trivial'' vectors of the form $q$ times a standard basis vector have too large of an infinity-norm to be a solution to $\SIS_{na,nb,\xi}$ when $\xi < q$, as will be the case for our parameter choices.} 
\begin{equation}\label{eq:sis_euclidean_length}
    2^{2 \sqrt{na \log(q) \log(\delta(\mu)) }}.
\end{equation}
We assume that the entries of $v$ have the same magnitudes since a similar assumption is made in \cite[Appendix C.3]{Dilithium}. Then, to solve $\SIS_{na,nb,\xi}$, we set the BKZ block-size $\mu$ to be the smallest positive integer such that
\begin{equation}
    \frac{1}{\sqrt{na + nb}}\cdot  2^{2 \sqrt{na \log(q) \log(\delta(\mu)) }} \leq \xi.
\end{equation}

\subsection{Parameter sets for different security levels}

To set $\Dilithium$ parameters, we also require $q = 1 \modulo 2\gamma_2$, $q> 4\gamma_2$ (see \cite[Lemma 1]{Dilithium} or \cite[Lemma 4.1]{Concrete}), and $\beta = \tau \eta$ (see \cite[Table 2]{Dilithium}). Moreover, we set parameters to minimize the following metrics \cite{Concrete}:
\begin{enumerate}
    \item the public key size in bytes: $(nk(\ceil{\log(q)} -d) + 256)/8$; 
    \item the signature size in bytes: $(n  l  \ceil{\log(2\gamma_1)} + nk + \tau (\log(n) + 1))/8$;
    \item the expected number of repeats to sign a message: $\exp\bigl(n\beta\bigl(\frac{l}{\gamma_1} + \frac{k}{\gamma_2}\bigr)\bigr)$.
\end{enumerate}

In \cref{tab:compare_dilithium_specs,tab:compare_dilithium_qrom,tab:noComp}, we give parameter sets achieving different levels of security that we calculated using the methodology described above. In all tables, we use:
\begin{equation}\label{eq:modulus}
    q_0 \coloneqq 12439554041857 = 2^{11} \cdot 3 \cdot 19\cdot 1447\cdot 73643 + 1 \approx 2^{43.5}.
\end{equation}
In particular, $q_0 = 1 \modulo 2n$.

Having established our attack model, we quantify the security provided by the proposed parameter sets for both $\Dilithium$ \cite{Dilithium} and $\DilithiumQROM$ \cite{Concrete} using our model in \cref{tab:compare_dilithium_specs,tab:compare_dilithium_qrom}. In those tables, we also provide new parameter sets that guarantee the same security if we analyzed the security of $\STMSIS$ using \cref{thm:stmsis}, in particular, \cref{eq:reduction_loss}. The new parameter sets are chosen in a way that minimizes their corresponding public key and signature sizes, as well as the expected number of repeats in $\Sign$. In \cref{tab:noComp}, we provide our recommended parameter sets at the five security levels specified by NIST.\footnote{The headings ``SL$l$'' appearing in \cref{tab:compare_dilithium_specs} follow the headings used in \cite[Table 2]{Dilithium}. Under our attack model, they do not exactly correspond to the desired security of NIST's SL$l$. This explains the need for \cref{tab:noComp} and why \cref{tab:noComp} differs from \cref{tab:compare_dilithium_specs}.}

\begin{table}[ht]
    \centering
    \scalebox{0.8}{
    \renewcommand*{\arraystretch}{1.1}
    \begin{tabular}{c|c|c|c|c|c|c}
        \hline
        &\multicolumn{3}{c|}{$\Dilithium$ \cite{Dilithium}} &\multicolumn{3}{c}{Our work} \\
        &SL2 &SL3 &SL5 &SL2 &SL3 &SL5 \\
        \hhline{=|=|=|=|=|=|=}
        $q$ &$2^{23}-8191$ &$2^{23}-8191$ &$2^{23}-8191$ &$q_0$ &$q_0$ &$q_0$ \\
        $n$ &256 &256 &256 &512 &512 &512 \\
        $(k, l)$ &(4, 4) &(6, 5) &(8, 7) &(10, 4) &(12, 8) &(16, 13) \\
        $d$ &13 &13 &13 &15 &15 &15 \\
        $\tau$ &39 &49 &60 &40 &40 &40\\
        $\gamma_1$ &$2^{17}$ &$2^{19}$ &$2^{19}$ &220929 &370432 &555648 \\
        $\gamma_2$ &95232 &261888 &261888 &441858 &740864 &1111296 \\
        $\zeta$ &350209 &724481 &769537 &1539077 &2137089 &2877953 \\
        $\zeta'$ &380930 &1048184 &1048336 &1767434 &2963458 &4445186\\
        $\eta$ &2 &4 &2 &2 &2 &2\\
        $\eta'$ &N/A &N/A &N/A &8 &4 &2 \\
        \hline
        $pk$ size (bytes)       &1312 &1952 &2592 &18592 &22304 &29728 \\
        $\sigma$ size (bytes)   &2476 &3448 &4804 &5554 &11058 &18546\\
        Expected Repeats        &4.25 &5.10 &3.85 &5.30 &4.70 &4.70 \\
        \hline
        $\LWE$ $\BKZ$ Block-Size    &448 &669 &911 &605 &1205 &2111\\
        Quantum Core-SVP            &118 &177 &241 &160  &319 &559\\
        \hline
        ``$\STMSIS$'' $\BKZ$ Block-Size         &N/A &N/A &N/A &1753 &2177 &3025\\
        Quantum Core-SVP            &N/A &N/A &N/A &100   &141  &205\\
        \hline
        $\SIS$ $\BKZ$ Block-Size    &363 &533  &773 &4942  &5644  &7423\\
        Quantum Core-SVP            &96  &141  &204 &1309  &1495  &1967\\
        \hline
    \end{tabular}}
    \caption{We give parameter sets that matches the quantized security of those proposed in the $\Dilithium$ specifications \cite{Dilithium}. $q_0$ is defined in \cref{eq:modulus}. The ``$\STMSIS$'' block-size should be understood as the block-size of the $\LWE$ problem reduced to via \cref{thm:stmsis} and \cref{sec:MLWE}.}
    \label{tab:compare_dilithium_specs}
\end{table}

Compared to the original $\Dilithium$ at ``SL3'', we find an increase in public key size of $\approx 11.4\times$ and an increase in signature size of $\approx 3.2\times$ \cite{Dilithium}. However, our results are provably secure based on conventional hardness assumptions for the $\MSIS$ and $\MLWE$ problems, 
whereas $\Dilithium$ must also assume that $\STMSIS$ is hard for the parameters that they set.  (See the discussion in \cref{subsec:knownsecurity}.)  Therefore, the main advantage of our parameters compared to $\Dilithium$ is that ours are based on rigorous reductions from hard lattice problems, whereas $\Dilithium$'s are based on highly heuristic reductions. We note that the heuristic reduction from $\STMSIS$ to (a variant of) $\MSIS$ given in \cite[End of Section 6.2.1]{Dilithium} has been recently challenged~\cite{STMSISConcrete}.

Compared to $\DilithiumQROM$ at its recommended security level, we find an increase in public key size of $\approx 2.9\times$ and an increase in signature size of $\approx 1.3\times$ \cite{Concrete}. However, while both parameter sets produce schemes that can be proven secure under the assumptions that $\MSIS$ and $\MLWE$ are hard, our parameter sets allow the use of the NTT and are therefore more efficient to implement than those of $\DilithiumQROM$. We analyze this difference in greater detail below.

\begin{table}[ht]
        \centering
        \scalebox{0.85}{
        \renewcommand*{\arraystretch}{1.1}
        \begin{tabular}{c|c|c|c|c}
            \hline
            & \multicolumn{2}{c|}{$\DilithiumQROM$ \cite{Concrete}} &\multicolumn{2}{c}{Our work} \\
             &recommended &very high &recommended &very high \\
            \hhline{=|=|=|=|=}
            $q$  &$2^{45}-21283$ &$2^{45}-21283$ &$q_0$ &$q_0$ \\
            $n$ &512 &512 &512 &512 \\
            $(k, l)$ &(4, 4) &(5, 5) &(12, 5) &(13, 8) \\
            $d$ &15 &15 &15 &15 \\
            $\tau$ &46 &46 &40 &40 \\
            $\gamma_1$ &905679 &905679 &279949 &370432 \\
            $\gamma_2$ &905679 &905679 &555648 &740864 \\
            $\zeta$ &2565023 &2565023 &1766657 &2137089\\
            $\zeta'$ &3622718 &3622718 &2222594 &2963458 \\
            $\eta$ &7 &3 &2 &2 \\
            $\eta'$ &N/A &N/A &5 &4 \\
            \hline
            $pk$ size (bytes) &7712 &9632 &22304 &24160 \\
            $\sigma$ size (bytes) &5690 &7098 &7218 &11122 \\
            Expected Repeats &4.29 &2.18 &5.03 &4.97 \\
            \hline
            $\LWE$ $\BKZ$ Block-Size    &499 &620 &794     &1232 \\
            Quantum Core-SVP            &132 &164 &210     &326  \\
            \hline
            ``$\STMSIS$'' $\BKZ$ Block-Size          &N/A &N/A &2118    &2374 \\
            Quantum Core-SVP            &N/A &N/A &133     &167  \\
            \hline
            $\SIS$ $\BKZ$ Block-Size    &N/A &N/A &5910    &6197 \\
            Quantum Core-SVP            &N/A &N/A &1566    &1642 \\
            \hline
        \end{tabular}
        }
        \caption{We give parameter sets that match the quantized security of those proposed in $\DilithiumQROM$ \cite{Concrete}. $q_0$ is defined in \cref{eq:modulus}. In the ``Our work'' columns, we assume $Q$ is bounded by $2^{96}$, which corresponds to NIST Security Level 3. The ``$\STMSIS$'' block-size should be understood as the block-size of the $\LWE$ problem reduced to via \cref{thm:stmsis} and \cref{sec:MLWE}.}
        \label{tab:compare_dilithium_qrom}
\end{table}

The main reason why we must increase the public key and signature sizes is due to the loss in the reduction from $\MLWE$ to $\STMSIS$, as stated in \cref{thm:stmsis}. Concretely, the loss manifests as \cref{eq:reduction_loss}, which we used to calculate the Quantum Core-SVP numbers for the $\STMSIS$-based $\MLWE$ problem. An interesting open question is to understand whether this loss is intrinsic.

\begin{table}[ht]
    \centering
    \scalebox{0.85}{
    \renewcommand*{\arraystretch}{1.1}
    \begin{tabular}{c|c|c|c|c}
        \hline
        &SL1 &SL2 &SL3 &SL4/5 \\
        \hhline{=|=|=|=|=}
        $q$ &$q_0$ &$q_0$ &$q_0$ &$q_0$ \\
        $n$ &512 &512 &512 &512 \\
        $(k, l)$ &(7, 7) &(9, 9) &(10, 10) &(13, 13) \\
        $d$ &15 &15 &15 &15 \\
        $\tau$ &40 &40 &40 &40 \\
        $\gamma_1$ &277824 &329916 &370432 &555648 \\
        $\gamma_2$ &555648 &659832 &740864 &1111296 \\
        $\zeta$ &1766657 &1975025 &2137089 &2877953 \\
        $\zeta'$ &2222594 &2639330 &2963458 &4445186 \\
        $\eta$ &2 &2 &2 &2 \\
        $\eta'$ &7 &5 &4 &2 \\
        \hline
        $pk$ size (bytes)       &13024 &16736 &18592 &24160 \\
        $\sigma$ size (bytes)   &9458  &12146 &13490 &18354 \\
        Expected Repeats        &4.70  &5.34 &5.25 &4.21 \\
        \hline
        $\LWE$ $\BKZ$ Block-Size    &967 &1325 &1509 &2079 \\
        Quantum Core-SVP            &256 &351 &399  &550 \\
        \hline
        ``$\STMSIS$'' $\BKZ$ Block-Size          &1252 &1665 &1866 &2454 \\
        Quantum Core-SVP            &66   &88   &100   &130  \\
        \hline
        $\SIS$ $\BKZ$ Block-Size    &3100  &4064 &4525  &5822  \\
        Quantum Core-SVP            &821   &1076 &1199  &1542  \\
        \hline
    \end{tabular}}
    \caption{We give parameter sets that most closely match the security levels requested by NIST \cite[Appendix A]{FIPS204}. $q_0$ is defined in \cref{eq:modulus}. The ``$\STMSIS$'' block-size should be understood as the block-size of the $\LWE$ problem reduced to via \cref{thm:stmsis} and \cref{sec:MLWE}.}
    \label{tab:noComp}
    \end{table}

Next, we quantitatively compare the efficiency of ring multiplication for the parameter sets in  \cref{tab:compare_dilithium_qrom}. Our work uses $q=q_0$ and $n=512$. Since $q = 1 \mod 2n$, we can multiply two elements in $R_q$ using the NTT, which uses $\frac{3}{2}n\log(n) + 2n = 7936$ multiplications in $\mathbb{Z}_q$ and $3n\log(n) = 13824$ additions in $\mathbb{Z}_q$ \cite[Section 2.2]{LiangNTT2021}. 

$\DilithiumQROM$ uses a $q$ such that $q=5\mod 8$ and we can no longer use the NTT to multiply elements in $R_q$. Instead, we consider the Hybrid-NTT (H-NTT) \cite[Section 5]{LiangNTT2021}. When $q = 1 \mod (n/2^{\alpha+\beta-1})$, where $\alpha,\beta$ are non-negative integer parameters, and $n$ is a power of $2$, H-NTT can multiply two elements in $R_q$ using
\begin{equation}
    \frac{3}{2} n\log(n) + \Bigl(3\cdot 2^{\alpha+\beta-3} + 2^{\alpha-2} + 3\cdot 2^{\beta-3} + 2^{\alpha-\beta-2}-\frac{3}{2}(\alpha+\beta)+\frac{5}{4}\Bigr)n
\end{equation}
multiplications in $\mathbb{Z}_q$, and
\begin{equation}  
3n\log(n)+ \Bigl(5\cdot 2^{\alpha+\beta-2} + 5\cdot 2^{\beta-2} + 5\cdot 2^{\alpha-2} - 3(\alpha+\beta) - \frac{15}{4} \Bigr)n
\end{equation}
additions in $\mathbb{Z}_q$. $\DilithiumQROM$ uses  $q = 2^{45}-21283$ and $n=512$ so the condition  $q = 1 \mod (n/2^{\alpha+\beta-1})$ requires $\alpha+\beta \in \{8, 9, 10\}$.\footnote{Note that $2^{\alpha+ \beta-1}$ needs to be $n$, $n/2$, or $n/4$ for $q = 1 \mod (n/2^{\alpha+\beta-1})$ to be compatible with $q =  5 \mod 8$. This means H-NTT would use $\Omega(n^2)$ multiplications and additions in $\mathbb{Z}_q$ when multiplying elements of $R_q$ in $\DilithiumQROM$.} The number of multiplications and additions in $\mathbb{Z}_q$ is minimized by setting $\alpha=\beta=4$. Therefore, H-NTT uses $\frac{3}{2} n \log(n) + 95.5 n = 55808$ multiplications and $3n \log(n) + 332.25n = 183936$ additions in $\mathbb{Z}_q$ per ring element multiplication.

We proceed to count the number of ring element multiplications and additions used by $\Dilithium$'s algorithms $(\KeyGen, \Sign, \Verify)$ when $\Sign$ involves $r \in \mathbb{N}$ repeats. For the count, we use the simplified descriptions of these algorithms given in \cref{fig:dilithium_simplified}.
\begin{alignat*}{5}
    &\text{Multiplications} \quad\quad&&\KeyGen\colon kl, \quad &&\Sign\colon (kl+k+l)r, \quad &&\text{and} \quad
    &&\Verify\colon kl+k
    \\
  &\text{Additions} \quad\quad &&\KeyGen\colon kl, \quad &&\Sign\colon (kl+l)r,
    \quad &&\text{and} \quad &&\Verify\colon kl
\end{alignat*}
Note that adding two ring elements requires $n$ additions in $\mathbb{Z}_q$.

Now, in \cref{tab:operationCounts}, we compare the number of multiplications and additions in $\mathbb{Z}_q$ used by $\Dilithium$ when instantiated with the parameter sets in \cref{tab:compare_dilithium_qrom}.

\begin{table}[ht]
        \centering
        \scalebox{0.85}{
        \renewcommand*{\arraystretch}{1.1}
        \begin{tabular}{c|c|c|c|c}
            \hline
            &\multicolumn{2}{c|}{$\DilithiumQROM$ \cite{Concrete}} &\multicolumn{2}{c}{Our work} \\
            &recommended &very high &recommended &very high \\
            \hline
            \multicolumn{5}{c}{Multiplications in $\mathbb{Z}_q$} \\
            \hline 
            Gen     &892928   &1395200  &476160  &825344   \\
            Sign    &5745992  &4258150  &3073692  &4930240 \\
            Verify  &1116160  &1674240  &571392  &928512  \\
            \hline
            \multicolumn{5}{c}{Additions in $\mathbb{Z}_q$} \\
            \hline
            Gen     &2951168  &4611200  &860160  &1490944  \\
            Sign    &18981980 &14067802 &5521572 &8873160 \\
            Verify  &3686912  &5530880  &1026048  &1670656  \\
            \hline
        \end{tabular}
        }
        \caption{The number of $\mathbb{Z}_q$ additions and $\mathbb{Z}_q$ multiplications required to implement the ring operations performed by the $\KeyGen, \Sign$, and $\Verify$ algorithms of $\Dilithium$. These numbers are calculated using the parameters in \cref{tab:compare_dilithium_qrom} and the analysis from \cite{LiangNTT2021}. }
        \label{tab:operationCounts}
\end{table}

\cref{tab:operationCounts} shows that $\DilithiumQROM$ at its recommended security level would require approximately the following increases in $\mathbb{Z}_q$-operation counts when compared to our work:
\begin{alignat*}{5}
    &\text{Multiplications} \quad\quad  &&\KeyGen\colon 1.9 \times, \quad &&\Sign\colon 1.9 \times, \quad &&\text{and} \quad
    &&\Verify\colon 2.0\times
        \\
    &\text{Additions} \quad\quad &&\KeyGen\colon 3.4\times, \quad &&\Sign\colon 3.4\times,
    \quad &&\text{and} \quad &&\Verify\colon 3.6\times
\end{alignat*}

We therefore identify a cost-benefit trade-off between the two provably secure formulations of $\Dilithium$, our work and $\DilithiumQROM$, at the recommended security level. Our work's public key and signature sizes are $2.9\times$ and $1.3\times$ larger than $\DilithiumQROM$'s, respectively. However, our scheme requires $1.9\times$ to $3.6\times$ fewer $\mathbb{Z}_q$-operations to implement. Moreover, unlike $\DilithiumQROM$, our work proves security on $\Dilithium$'s native ring where $q=1\mod 2n$.

We make a final remark on the concrete security analysis of our work as well as those originally done for $\Dilithium$ and $\DilithiumQROM$: no analysis accounts rigorously for potential differences in the hardness between $\LWE$ with a uniform error distribution and $\SIS$ under the $\ell_\infty$ norm as compared to the better-studied versions of these problems which employ a Gaussian error distribution and the $\ell_2$ norm, respectively. However, the hardness of the former problems is comparable to the hardness of the latter problems over parameter regimes that are polynomially related in the security parameter \cite{LatticeNorms,SmallParams}. Therefore, like the original analyses of $\Dilithium$ and $\DilithiumQROM$, we assume that the differences in hardness are not significant enough to seriously threaten security.

\section{Acknowledgments}

This work was supported by the National Institute of Standards and Technology (NIST) and the Joint Center for Quantum Information and Computer Science (QuICS) at the University of Maryland. This research paper is not subject to copyright in the United States.  The opinions, findings, and conclusions in the paper are those of the authors and do not necessarily reflect the views or policies of NIST or the United States Government.

We thank Marcel Dall'Agnol, Jiahui Liu, Yi-Kai Liu,  and Ray Perlner for helpful feedback and correspondence. We thank Amin Shiraz Gilani for his involvement during the early stages of this project. 

\printbibliography
\end{document}